\documentclass{article}
\usepackage{arxiv}

\usepackage{times}

\usepackage{soul}
\usepackage{url}
\usepackage[utf8]{inputenc}
\usepackage[small]{caption}
\usepackage{graphicx}
\usepackage{amsmath}
\usepackage{booktabs}
\usepackage{algorithm}
\usepackage{algorithmic}
\usepackage{bm}
\usepackage{amssymb}
\usepackage{amsfonts}
\usepackage{xcolor}
\usepackage{amsthm}
\usepackage{subcaption}
\usepackage{multirow}
\usepackage{comment}
\usepackage{dcolumn}
\usepackage{natbib}
\setcitestyle{authoryear,open={(},close={)}}

\theoremstyle{definition}
\newtheorem{example}{Example}[section]
\theoremstyle{plain}
\newtheorem{assumption}{Assumption}[section]

\newtheorem{theorem}{Theorem}[section]
\newtheorem{lemma}{Lemma}[section]
\newtheorem{corollary}{Corollary}[section]
\theoremstyle{remark}
\newtheorem{remark}{Remark}[section]
\newcommand\independent{\protect\mathpalette{\protect\independenT}{\perp}}
\def\independenT#1#2{\mathrel{\rlap{$#1#2$}\mkern2mu{#1#2}}}

\title{Orthogonal Series Estimation for the Ratio of Conditional Expectation Functions}

\author{
Kazuhiko Shinoda\\
Graduate School of Economics, Keio University\\
\And
Takahiro Hoshino\\
Faculty of Economics, Keio University\\
RIKEN AIP
}
\date{\empty}

\begin{document}
\maketitle
\begin{abstract}
In various fields of data science, researchers are often interested in estimating the ratio of conditional expectation functions (CEFR).
Specifically in causal inference problems, it is sometimes natural to consider ratio-based treatment effects, such as odds ratios and hazard ratios, and even difference-based treatment effects are identified as CEFR in some empirically relevant settings.
This chapter develops the general framework for estimation and inference on CEFR, which allows the use of flexible machine learning for infinite-dimensional nuisance parameters.
In the first stage of the framework, the orthogonal signals are constructed using debiased machine learning techniques to mitigate the negative impacts of the regularization bias in the nuisance estimates on the target estimates.
The signals are then combined with a novel series estimator tailored for CEFR.
We derive the pointwise and uniform asymptotic results for estimation and inference on CEFR, including the validity of the Gaussian bootstrap, and provide low-level sufficient conditions to apply the proposed framework to some specific examples.
We demonstrate the finite-sample performance of the series estimator constructed under the proposed framework by numerical simulations.
Finally, we apply the proposed method to estimate the causal effect of the 401(k) program on household assets.
\end{abstract}

\section{Introduction}
In various fields of data science, researchers often face problems of estimating \emph{the ratio of conditional expectation functions (CEFR)}.
Although CEFR appears not only in causal inference studies, there are several important examples of CEFR in the treatment effect estimation literature.
When an outcome of interest is the relapse rate of a specific disease, it is more natural to consider the ratio of the conditional expectation of potential outcomes $E[Y_1|X]/E[Y_0|X]$, rather than the difference $E[Y_1|X]-E[Y_0|X]$, as a measure of causal effects of a medical treatment, where $Y_1$ and $Y_0$ are the potential outcome with and without a treatment, respectively, and $X$ is a vector of baseline covariates.
Likewise, ratio-based treatment effects such as the odds ratio and hazard ratio have been widely used especially in clinical settings.
Furthermore, in a data combination setting where an outcome and a treatment status are only separately observed, both conditional average treatment effect (CATE) and local average treatment effect (LATE) are identified as in the form of CEFR \citep{yamane,shinoda2022estimation}, while these effects are defined as the difference of the potential outcomes.

In this article, we start by developing a novel series estimator for CEFR in a very simple setting without selection bias in observed data.
This series estimator is itself useful in estimating treatment effects when data can be collected completely at random from the population of interest, but such randomized data are often not available in practice.
Technically, when there is selection bias in collected data, we need to estimate potentially infinite-dimensional nuisance parameters to adjust for the bias, but these parameters may be hard to estimate with a ``sufficiently high quality'' in observational studies on complex systems since they can be very high-dimensional and/or highly nonlinear.
The highly complex nuisance parameters do not satisfy the traditional assumptions that limit the complexity of a function class, and therefore the resulting semiparametric estimator fails to be $\sqrt{N}$-consistent.
We employ debiased machine learning (DML), a set of techniques to enable the use of flexible machine learning (ML) methods for nuisance estimation, to develop a simple and general framework for constructing a high-quality estimator for CEFR even in the presence of selection bias in observed data.

The major contribution of this study is the development of a novel inference framework for CEFR with theoretical guarantees.
We derive the general asymptotic results for estimation and uniform inference on the best linear approximation to the target CEFR under the proposed framework, including the validity of the Gaussian bootstrap.
It is worth noting that we do not have to add stronger regularity conditions than the assumptions previously known in the literature to establish the theoretical results in this article.
In addition to the general results, we provide a set of low-level sufficient conditions to apply the proposed framework to several specific settings.
Besides the asymptotic analysis, we conduct numerical simulations to evaluate the performance of the proposed method on finite samples.
We also illustrate the use of the framework in an empirical example.

This study builds upon three important bodies of research within the semiparametric literature: DML, series estimation and CEFR estimation.
DML \citep{chernozhukov2018,chernozhukov2022locally} enables inference on the finite-dimensional parameter in the presence of infinite-dimensional nuisance parameters by using the Neyman orthogonal moment conditions.
When the moment condition satisfies the Neyman orthogonality \citep{neyman1959optimal}, bias in the nuisance estimates has no first-order asymptotic effect on the estimator of the target parameter.
DML is powerful enough to allow the use of a broad class of ML estimators ---such as $\ell_1$-penalized methods in sparse models \citep{bickel2009simultaneous,buhlmann2011statistics,belloni2011square,belloni2012sparse,belloni2011ell1,belloni2013least}, $L_2$-boosting methods in sparse linear models \citep{luo2016high}, neural nets \citep{schmidt2020nonparametric,farrell2021deep,kohler2021on}, trees and random forest \citep{wager2015adaptive,syrgkanis2020estimation}--- under a broad range of data generating processes and for various causal parameters.
Moreover, the extensions of DML have been proposed to estimate a function, e.g. CATE or continuous treatment effect, with the existence of complex infinite-dimensional nuisance parameters \citep{jacob2019group,zimmert2019nonparametric,colangelo2020double,semenova2021debiased,fan2022estimation}.
The present study provides a novel method for the orthogonal estimation under a variety of realistic settings that cannot be covered by the previous works, such as ratio-based treatment effects, and treatment effects in the data combination settings.

The second foundation of this study is the least squares series estimation.
Series estimation is a type of nonparametric estimation method that approximates a function of interest by a linear combination of multiple basis functions. 
It is especially useful when the exact functional form of the target function is unknown.
Its asymptotic properties have been investigated intensively in the literature \citep{van1990estimating,andrews1991asymptotic,eastwood1991adaptive,gallant1991asymptotic,newey1997convergence,van2002m,huang2003local,chen2007large,cattaneo2013optimal,belloni2015some,chen2015optimal}, among which we will mainly rely on the results in \cite{belloni2015some}.
This study contributes to the series estimation literature by establishing the asymptotic theory for the CEFR problems rather than the standard regression setting.

This study is also closely related to CEFR estimation in causal inference.
A simple yet important example is ratio-based treatment effects such as the odds and hazard ratio.
To the best of our knowledge, the previous works on the estimation of the ratio-based treatment effects all impose assumptions on the functional form somewhere in the model \citep{dukes2018note,liang2020relative,yadlowsky2021estimation,lee2022survival}.
For example, \cite{liang2020relative} supposes that the ratio-based treatment effects can be expressed as the monotone single index model, and \cite{yadlowsky2021estimation} imposes a stronger condition that the monotone link function is an exponential function.
On the other hand, this study considers a fully nonparametric model for treatment effects, imposing no assumptions on the functional form of CEFR.

Furthermore, \cite{yamane} and \cite{shinoda2022estimation} show that difference-based CATE and LATE are identified in the form of CEFR in the data combination setting where we cannot observe an outcome and a treatment status simultaneously in a single dataset, and develop estimation methods for CEFR.
This study accommodates the generalized version of their problem settings, where each separate dataset contains selection bias.
Their estimators cannot handle the selection bias without introducing additional nuisance parameters, but their methods are not orthogonal to the nuisance parameters.

The rest of this article is organized as follows.
In Section \ref{sec:direct}, we develop a simple series estimator for CEFR.
We propose the general framework for estimation and inference of CEFR in Section \ref{sec:framework}.
The main theoretical results for the asymptotic properties of the series estimator under the proposed framework are presented in Section \ref{sec:asymtheory}, and the application of the framework to some specific treatment effects problems is shown in Section \ref{sec:app_dml}.
In Section \ref{sec:simdml}, we illustrate the finite sample performance of the proposed method by numerical simulations.
We also apply the proposed method to the empirical example of estimating a causal effect of participation in 401(k) on net financial assets in Section \ref{sec:empirical}.
Finally, we provide further discussion on the limitation and future direction of the present study in Section \ref{sec:discussiondml}.

\section{Direct Series Estimator for CEFR}\label{sec:direct}
Consider the following model:
\begin{gather}
    \theta_0(X)=\frac{\nu_0(X)}{\zeta_0(X)},\hspace{5mm}U=\nu_0(X)+\varepsilon_U,\hspace{5mm} T=\zeta_0(X)+\varepsilon_T,\notag\\
    E[\varepsilon_U|X]=E[\varepsilon_T|X]=0,\label{eq:regmodel}
\end{gather}
where $X\in\mathcal{X}\subset\mathbb{R}^q$ is a $q$-dimensional vector of covariates, and $\zeta_0(x)\neq0$ for all $x\in\mathcal{X}$ is necessary for $\theta_0(x)$ to be well-defined.
Now we are interested in estimating $\theta_0(x)$ from separately observed iid samples: $\{u_i,x_i\}_{i=1}^{N_U}$ of $(U,X)$ and $\{t_i,x_i\}_{i=1}^{N_T}$ of $(T,X)$.
Suppose $N=N_U=N_T$ without loss of generality.
Also, assume that the samples $\{u_i,x_i\}_{i=1}^N$ and $\{t_i,x_i\}_{i=1}^N$ are complete random draws from the population of interest, i.e. there is no selection bias in these separate samples.
Then, regarding $U$ and $T$ as outcome variables with and without a treatment, respectively, this problem coincides with the estimation of ratio-based treatment effects \citep{ogburn2015,dukes2018note,yadlowsky2021estimation} in randomized control trials.
Moreover, estimation of difference-based CATE and LATE from separately observed samples \citep{yamane,shinoda2022estimation} is included in the model (\ref{eq:regmodel}).
Although the estimator developed below primarily aims at estimating $\theta_0(x)$ from the separate samples, it can also be applicable to the estimation from joint samples $\{u_i,t_i,x_i\}_{i=1}^N$ of $(U,T,X)$.
Thereafter, we omit arguments of functions if they are obvious from the context.

It follows from the model (\ref{eq:regmodel}) that $E[U-\zeta_0(X)\theta_0(X)|X]=0$.
Then, using this moment condition, we can consider the linear approximation to the target function $\theta_0(x)$ by a vector of $k$ basis functions $p(x):=(p_1(x),\ldots,p_k(x))'$ as $\theta(x)=p(x)'\beta$, where
\begin{align}
    \beta&=\underset{b\in\mathbb{R}^k}{\mathrm{argmin}}E[(U-\zeta_0(X)p(X)'b)^2]\notag\\
    &=E[p(X)p(X)'\zeta_0(X)^2]^{-1}E[p(X)\zeta_0(X)U].\label{eq:serest}
\end{align}
However, as the true denominator function $\zeta_0(x)$ is unknown, the above equation cannot be directly applied for estimation.
To deal with this issue, \cite{shinoda2022estimation} proposed to plug-in the estimator $\hat{\zeta}(x)$ of $\zeta_0(x)$, while \cite{yamane} introduced an auxiliary function and reformulate the problem into minimax optimization.
Both of these approaches have problems: the former lacks the theoretical justification to use flexible nonparametric methods for $\hat{\zeta}(x)$ since simply plugging-in nonparametric estimates generally does not ensure $\sqrt{N}$-consistency of the target estimator, and the latter suffers from unreliable hyperparameter selection caused by the minimax objective.
In what follows, we develop a simpler one-step estimator for CEFR in the spirit of \cite{vapnik1999nature}'s principle: \emph{When solving a given problem, try to avoid solving a more general problem as an intermediate step}.

Define $f(x):=\zeta_0(x)p(x)$ and its series estimator 
\begin{align*}
\hat{f}(x):=E[p(X)p(X)'T]E[p(X)p(X)']^{-1}p(x)=:QS^{-1}p(x).
\end{align*}
Substituting $f$ for (\ref{eq:serest}) gives
\begin{align*}
    \beta=E[f(X)f(X)']^{-1}E[f(X)U].
\end{align*}
Then, by replacing $f(x)$ with its estimator $\hat{f}(x)$, we obtain
\begin{align*}
    \beta&=(QS^{-1}SS^{-1}Q)^{-1}QS^{-1}E[p(X)U]\\
    &=Q^{-1}SQ^{-1}QS^{-1}E[p(X)U]\\
    &=Q^{-1}E[p(X)U],
\end{align*}
and its sample analog estimator $\hat{\beta}=\hat{Q}^{-1}E_N[p(x_i)u_i]$, where $E_N[g(x_i)]:=N^{-1}\sum_{i=1}^N f(x_i)$ for some function $g$ and $\hat{Q}:=E_N[p(x_i)p(x_i)'t_i]$.
The key observation here is that we can compute $\hat{\beta}$ even when $\{u_i,x_i\}_{i=1}^N$ and $\{t_i,x_i\}_{i=1}^N$ are only separately observed.
Henceforth, we refer to this estimator as the \emph{Direct Series Ratio (DSR)} estimator.

\begin{remark}[Asymptotic Properties of the DSR Estimator]
Although the estimator $\hat{f}(x)$ is plugged-in during the derivation process of DSR, there is no actual need to calculate $\hat{f}(x)$.
Using the nonparametric estimator of nuisance functions as done in \cite{shinoda2022estimation} may adversely affect the theoretical and practical properties of a target estimator since nonparametric estimators may have biases due to model selection and regularization.
DSR, on the other hand, does not require the estimated value.
Consequently, it converges at the same rate as in the general regression setting under mild conditions.
The pointwise and uniform asymptotic theory for the DSR estimator immediately follows from the theoretical results shown in \ref{sec:asymtheory}.
\end{remark}

\begin{remark}[Regularized Estimation]
For practical purposes, the estimation of $\hat{Q}$ can be unstable when the sample size is small or when $\hat{Q}$ is close to a singular matrix.
One way to stabilize the estimation is to perform $\ell_2$-regularization (ridge estimation) with $\lambda$ as the regularization parameter:
\begin{align}
    \tilde{\beta}=(\hat{Q}+\lambda I_k)^{-1}E_N[p(x_i)u_i].\label{eq:ridge_param}
\end{align}
The impact of $\lambda$ on the performance of DSR is investigated in the simulation study in Section \ref{sec:dsrsim}.
\end{remark}

\begin{remark}[Model Selection]\label{rm:cv}
To implement the estimators $\hat{\beta}$ and $\tilde{\beta}$, we need to choose the series length $k$ and the regularization parameter $\lambda$.
One practical way of selection is to find hyperparameters that minimize a criterion computed based on cross-validation (CV).
The Mean Squared Error (MSE) is a general criterion, but it is difficult to evaluate MSE directly in the setting (\ref{eq:regmodel}).
However, in special cases where $\zeta_0(x)>0$ for all $x\in\mathcal{X}$, it is possible to calculate a CV criterion that preserves the rank order of MSE.
Expanding MSE gives
\begin{align*}
    E[(U-\zeta_0(X)\hat{\theta}(X))^2]=E[U^2]-2E[U\zeta_0(X)\hat{\theta}(X)]+E[\zeta_0(X)^2\hat{\theta}(X)^2],
\end{align*}
and we can ignore the first term as it does not involve the estimator $\hat{\theta}$.
Moreover, if $\zeta_0(x)$ is positive for all $x\in\mathcal{X}$, 
\begin{align*}
    E[\zeta_0(X)\hat{\theta}(X)^2]-2E[U\hat{\theta}(X)]
\end{align*}
is monotone increasing in MSE.
Therefore, we can use the sample analogue as the criterion:
\begin{align}
    E_N[t_i\hat{\theta}(x_i)^2]-2E_N[u_i\hat{\theta}(x_i)].\label{eq:cvcriterion}
\end{align}
CV in the general situation where $\zeta_0$ can take negative values is a subject for future work.
This limitation, however, does not detract significantly from the value of this study since researchers often have \emph{a priori} knowledge that $\zeta_0(x)$ is positive in many practical situations.
Some of these situations are explained in Section \ref{sec:app_dml}.
\end{remark}

\section{General Framework for Estimation and Inference of CEFR}\label{sec:framework}
In this section, we propose a general framework for the CEFR problems by using DSR developed in the previous section as one of the building blocks.
The proposed framework is based on the generalized version of the model (\ref{eq:regmodel}), and therefore it can accommodate a variety of real-world applications.

\subsection{Setup}
In observational studies of complex systems, possibly high-dimensional covariates $X$ may be necessary for conditional independence to hold, which is one of the critical conditions in many causal inference problems.
However, researchers are often not interested in estimating treatment effects as a function of all covariates, but they need to focus on the heterogeneous relationships between treatment effects and only a few important factors.
In such a case, estimating the target function on the full vector $X$ is unnecessarily hard due to the curse of dimensionality, and directly estimating a function of a subvector can stabilize the estimation.
Therefore, suppose that our parameter of interest $\theta_0$ is now a function of $V\in\mathcal{V}$, a subvector of covariates $X$.
Also, suppose that $\nu_0$ and $\zeta_0$ in (\ref{eq:regmodel}) are defined as
\begin{align}
    \nu_0(v)=E[U(O,\eta_0)|V=v],\hspace{5mm}\zeta_0(v)=E[T(O,\eta_0)|V=v],\label{eq:signals}
\end{align}
where the random variables $U$ and $T$ depend on a vector of observed random variables $O$ and a set of infinite-dimensional nuisance parameters $\eta_0(x)$.
Note that $\eta_0$ is a function of $X$, reflecting that $X$ is indispensable for adjusting for selection biases.
We refer to $U$ and $T$ as signals, following \cite{semenova2021debiased}, and use the notations $U_0:=U(O,\eta_0)$ and $T_0:=T(O,\eta_0)$.

Among many possible choices of the signals $U$ and $T$ that satisfy (\ref{eq:signals}), we focus on the signals with the Neyman orthogonal property, which is defined later.
The Neyman orthogonality is the key property to deliver high-quality inference on the target function $\theta_0(v)$ even when modern ML estimators that do not satisfy Donsker conditions are used for the nuisance functions $\eta_0$.
The classic semiparametric theories ensure that the target estimator achieves $\sqrt{N}$-consistency by bounding the complexity of the nuisance space, but estimators in such a restricted class are not suitable for fitting high-dimensional and/or highly nonlinear nuisance functions in data-rich environments.
On the other hand, modern ML estimators are so flexible to fit any function that there is no need to specify the functional form of the nuisance parameters in advance.
ML estimators also perform well in high-dimensional settings by employing regularization to reduce variance at the cost of regularization bias.
These properties are especially important in observational studies in which the distribution of an outcome is a complex mixture of treated and non-treated populations, and there are many confounding factors.

Although ML estimators are effective in predicting the values of nuisance functions, their prediction is inherently biased by model selection and regularization.
Therefore, the naive plug-in estimator of the target function fails to be $\sqrt{N}$-consistent.
To ensure the desirable theoretical properties of the target estimator while using flexible ML estimators for nuisance functions, we need signals that are locally insensitive to the bias of the nuisance estimators.
Formally, the Neyman orthogonality is defined in terms of the pathwise (Gateaux) derivative as follows:
\begin{align*}
    \partial_s E[U(O,\eta_0+s(\eta-\eta_0))-\nu_0(V)|V=v]|_{s=0}&=0,\\
    \partial_s E[T(O,\eta_0+s(\eta-\eta_0))-\zeta_0(V)|V=v]|_{s=0}&=0,
\end{align*}
for all $v\in\mathcal{V}$, where $\partial$ is the partial derivative operator and $s\in[0,1]$ is a constant.

\subsection{Examples}\label{sec:dmlexample}
We describe some empirically relevant examples for the generalized setting.
The detailed discussions on conditions for identification and inference of estimands in each example are given in Section \ref{sec:app_dml}.

\begin{example}[Local Average Treatment Effect]\label{ex:late}
Let $Y_1$ and $Y_0$ be the potential outcomes realized only when an individual is treated and not treated, respectively.
Similarly, let $D_1$ and $D_0$ be the potential treatment status realized only when an individual is assigned to a treatment group or not.
We also denote a binary treatment assignment as $Z$.
The observable is $O=(Y,D,Z,X)$, where $D=ZD_1+(1-Z)D_0$ is a realized treatment status, and $Y=DY_1+(1-D)Y_0$ is a realized outcome.
The parameter of interest is LATE, which is defined as a treatment effect measured for the subpopulation of compliers
\begin{align*}
    \theta_0(v)=E[Y_1-Y_0|D_1>D_0,V=v],
\end{align*}
where compliers are individuals who always follow the given assignment $Z$.
LATE is often used when there is self-selection to receive a treatment (often referred to as noncompliance to treatment assignment), and thus the observed covariates $X$ are insufficient to adjust for selection bias in $D$.

A standard identification strategy for LATE is using $Z$ as an instrument to $D$.
Under several assumptions including conditions necessary for $Z$ to be a valid instrument, LATE is identified as
\begin{align*}
    \theta_0(v)=\frac{E[E[Y|Z=1,X]-E[Y|Z=0,X]|V=v]}{E[E[D|Z=1,X]-E[D|Z=0,X]|V=v]}.
\end{align*}
Therefore, we can consider LATE estimation as a problem of CEFR by setting $\nu_0(v)=E[\mu_0(1,X)-\mu_0(0,X)|V=v]$ and $\zeta_0(v)=E[\pi_0(1,X)-\pi_0(0,X)|V=v]$, where $\mu_0(z,x)=E[Y|Z=z,X=x]$ and $\pi_0(z,x)=E[D|Z=z,X=x]$.
\end{example}

\begin{example}[Ratio-Based Treatment Effects]\label{ex:rbte}
The observable vector is $O=(Y,D,X)$, where $Y$ is an outcome of interest, and $D$ is a binary treatment status.
As in the case of LATE, let $Y_1$ and $Y_0$ be the potential outcomes.
The parameter of interest is the ratio-based CATE:
\begin{align*}
    \theta_0(v)=\frac{E[Y_1|V=v]}{E[Y_0|V=v]},
\end{align*}
where $E[Y_0|V=v]\neq0$ for all $v\in\mathcal{V}$ is necessary for $\theta_0$ to be well-defined.
If unconfoundedness $Y_1,Y_0\independent D|X$, and other standard assumptions hold, the ratio-based CATE is identified as:
\begin{align*}
    \theta_0(v)=\frac{E[E[Y|D=1,X]|V=v]}{E[E[Y|D=0,X]|V=v]}.
\end{align*}
Therefore, we can consider the estimation of ratio-based CATE as a problem of CEFR by setting
\begin{align*}
    \nu_0(v)&=E[\mu_0(1,X)|V=v]=E[E[Y|D=1,X]|V=v],\\
    \zeta_0(v)&=E[\mu_0(0,X)|V=v]=E[E[Y|D=0,X]|V=v].
\end{align*}
\end{example}

\begin{example}[Instrumented Difference-in-Differences]\label{ex:idid}
Instrumented difference-in-Differences (IDID) is the method that combines the advantages of instrumental variables (IVs) and difference-in-differences \citep{Ye2020InstrumentedD,vo2022structural}.
IDID allows the identification of treatment effects under conditions milder than required by the instrumental methods or difference-in-differences alone.
Suppose that we observe a vector of random variables $O=(Y,D,Z,W,X)$, where $Y$ is an outcome of interest, $D$ is a binary treatment status, $Z$ is a binary treatment assignment, and $W$ is a binary time indicator.
Let $D_{zw}$ be the potential treatment status that would be observed if $Z=z$ in time $w$, and $Y_{dw}$ be the potential outcome that would be observed if $D=d$ in time $w$.
The parameter of interest is CATE
\begin{align*}
    \theta_0(v)=E[Y_1-Y_0|V=v],
\end{align*}
where we assume $E[Y_1-Y_0|V]=E[Y_{11}-Y_{01}|V]=E[Y_{10}-Y_{00}|V]$.
Under some identification assumptions, CATE is identified as:
\begin{align*}
    \theta_0(v)=\frac{E[\mu_0(1,1,X)-\mu_0(0,1,X)-\mu_0(1,0,X)+\mu_0(0,0,X)|V=v]}{E[\pi_0(1,1,X)-\pi_0(0,1,X)-\pi_0(1,0,X)+\pi_0(0,0,X)|V=v]},
\end{align*}
where $\mu_0(w,z,x)=E[Y|W=w,Z=z,X=x]$ and $\pi_0(w,z,x)=E[D|W=w,Z=z,X=x]$.
Therefore, we can consider the estimation of CATE in IDID as a problem of CEFR by setting
\begin{align*}
    \nu_0(v)&=\sum_{(w,z)\in\{0,1\}^2}(-1)^{w+z}E[\mu_0(w,z,X)|V=v],\\
    \zeta_0(v)&=\sum_{(w,z)\in\{0,1\}^2}(-1)^{w+z}E[\pi_0(w,z,X)|V=v].
\end{align*}
\end{example}

\begin{example}[Treatment Effects in the Data Combination Setting]\label{ex:dcte}
We extend the data combination setups studied in \cite{yamane} and \cite{shinoda2022estimation}.
Suppose that the parameter of interest is CATE or LATE, but we cannot observe an outcome $Y$ and a treatment status $D$ simultaneously.
Moreover, binary treatment assignment $Z$ is not observed at all.
We refer to a dataset that includes $Y$ as the outcome dataset, and a dataset that includes $D$ as the treatment dataset.
According to \cite{yamane} and \cite{shinoda2022estimation}, if we have two sets of the outcome and treatment datasets with different treatment regimes, we can identify CATE and LATE, where a treatment regime stands for the conditional treatment assignment probability $P(Z=1|X)$.
In summary, we have four different datasets in this setup: the outcome dataset with regime 1, outcome dataset with regime 0, treatment dataset with regime 1 and treatment dataset with regime 0.

Let $Y_d$ be the potential outcome that would be observed when $D=d$ and $D_z$ be the potential treatment status that would be observed when $Z=z$.
We also consider the potential treatment assignment $Z_w$ that would realize only in a dataset with regime $w=0,1$.
Suppose that the observable vector is $O=(HY+(1-H)D,W,H,X)$, where $W$ is a binary regime indicator, and $H$ is a binary dataset indicator.
Note that we do not have to observe $Z$ in this setting, but we must be sure that $P(Z_1|X=x)\neq P(Z_0|X=x)$.
\cite{yamane} and \cite{shinoda2022estimation} simplified the setup by assuming that samples in each dataset are completely random draws from the population of interest, formally
\begin{align*}
    Y_1,Y_0,D_1,D_0,Z_1,Z_0,X\independent H,W.
\end{align*}
However, it is not realistic that the four different datasets share the same joint distribution in observational studies.
Thus, we relax the above condition to the following conditional independence
\begin{align*}
    Y_1,Y_0,D_1,D_0,Z_1,Z_0\independent H,W|X.
\end{align*}
Under this condition and other corresponding identification assumptions, CATE and LATE are identified in the same form:
\begin{align*}
    \theta_0(v)=\frac{E[\mu_0(1,X)-\mu_0(0,X)|V=v]}{E[\pi_0(1,X)-\pi_0(0,X)|V=v]},
\end{align*}
where $\mu_0(w,x)=E[Y|H=1,W=w,X=x]$ and $\pi_0(w,x)=E[D|H=0,W=w,X=x]$.
Therefore, we can consider the estimation of CATE and LATE in the data combination setting as a problem of CEFR by setting $\nu_0(v)=E[\mu_0(1,X)-\mu_0(0,X)|V=v]$ and $\zeta_0(v)=E[\pi_0(1,X)-\pi_0(0,X)|V=v]$.
\end{example}

\subsection{Overall Inference Procedures}
We develop a two-stage estimator using the orthogonal signals and DSR proposed in Section \ref{sec:direct}.
We refer to this two-stage estimator as the \emph{Orthogonal Series Ratio (OSR)} estimator.
The first stage consists of constructing the orthogonal signals by cross-fitting.
Cross-fitting is another important technique to eliminate biases in the orthogonal signals constructed from finite samples.
The cross-fitting of the orthogonal signals is implemented as follows:
\begin{enumerate}
\item Let $\{J_g\}_{g=1}^G$ denote a $G$-fold random partition of the sample indices $[N]:=\{1,2,\ldots,N\}$, where $G$ is the number of partition. 
Suppose that the sample size of each fold $n:=N/G$ is an integer without loss of generality.
For each partition $g\in[G]$, define $J_g^c:=[N]\backslash J_g$.
\item For each partition $g\in[G]$, construct a set of estimators $\hat{\eta}_g:=\hat{\eta}(O_{i\in J_g^c})$ by using only the samples in $J_g^c$.
For any index $i\in J_g$, construct the orthogonal signals $\hat{U}_i:=U(O_{i\in J_g},\hat{\eta}_g)$ and $\hat{T}_i:=T(O_{i\in J_g},\hat{\eta}_g)$.
\end{enumerate}
Cross-fitting uses different samples for estimating the nuisance functions and constructing the orthogonal signals.
Such sample-splitting allows the nuisance estimators to be treated as non-random when constructing the orthogonal signals $\hat{U}$ and $\hat{T}$, which helps the debiasing of the target estimator.

In the second stage, the estimated orthogonal signals $\hat{U}$ and $\hat{T}$ are plugged-in to the DSR estimator:
\begin{align}
    \hat{\beta}=\hat{Q}^{-1}E_N[p_i\hat{U}_i],
\end{align}
where we redefine $Q$ and $\hat{Q}$ as $Q=E[pp'T_0]$ and $\hat{Q}=E_N[p_ip_i'\hat{T}_i]$, respectively.
The asymptotic covariance matrix of the OSR estimator is
\begin{align*}
    \Omega:=Q^{-1}E[pp'(\varepsilon_U+\theta_0\varepsilon_T)^2]Q^{-1}=\Omega_{k},
\end{align*}
where $\varepsilon_U:=U_0-\nu_0(V)$ and $\varepsilon_T:=T_0-\zeta_0(V)$ are the stochastic errors.
The sample analogue of the asymptotic variance is
\begin{align*}
    \hat{\Omega}:=\hat{Q}^{-1}E_N[p_ip_i'(\hat{U}_i-\hat{T}_i\hat{\theta}_i)^2]\hat{Q}^{-1}=\hat{\Omega}_{Nk},
\end{align*}
where $\hat{\theta}(v)=p(v)'\hat{\beta}$ is the OSR estimator of the target function $\theta_0$.

For statistical inference, denote the standard deviation of the OSR estimator at $v$ and its sample analogue as $\sigma(v)=\sqrt{p(v)'\Omega p(v)}$ and $\hat{\sigma}(v)=\sqrt{p(v)'\hat{\Omega}p(v)}$, respectively.
Then, we can write $t$-statistic as
\begin{align}
    \tau_N(v):=\frac{\hat{\theta}(v)-\theta_0(v)}{\hat{\sigma}(v)/\sqrt{N}},\label{eq:tstat}
\end{align}
and the bootstrapped $t$-statistic as
\begin{align}
    \hat{\tau}_N^b(v):=\frac{p(v)'\hat{\Omega}^{1/2}}{\hat{\sigma}(v)}\mathcal{N}_k^b,\label{eq:bootstrapt}
\end{align}
where $\mathcal{N}_k^b$ is a bootstrap draw from $N(0,I_k)$.
We will show that the Gaussian bootstrap is valid for the OSR estimator in the next section.
We can calculate the confidence bands for $\theta_0(v)$ as
\begin{align}
    [\underline{i}_N(v),\overline{i}_N(v)]:=[p(v)'\hat{\beta}-c_N(1-\delta)\hat{\sigma}(v)/\sqrt{N},p(v)'\hat{\beta}+c_N(1-\delta)\hat{\sigma}(v)/\sqrt{N}],\label{eq:confband}
\end{align}
where the critical value $c_N(1-\delta)$ is the $(1-\delta)$-quantile of $N(0,1)$ for the pointwise bands, and the $(1-\delta)$-quantile of $\sup_{v\in\mathcal{V}}|\hat{\tau}_N^b(v)|$ for the uniform bands.

\section{Main Theoretical Results}\label{sec:asymtheory}
Here, we present the main theoretical results on the asymptotic properties of the proposed OSR estimator.
These results heavily rely on the theoretical analyses in \cite{belloni2015some} and \cite{semenova2021debiased}.

First, we set up some additional notations.
Define a random variable that takes the value of the stochastic errors with the larger absolute value:
\begin{align*}
\varepsilon:=\underset{\{\varepsilon_{U},\varepsilon_{T}\}}{\mathrm{argmax}}\{|\varepsilon_{U}|,|\varepsilon_{T}|\},
\end{align*}
and the lower and upper bounds on its second moments:
\begin{align*}
    \underline{\varepsilon}^2(v):=\inf_{v\in\mathcal{V}}E[\varepsilon^2|V=v],\hspace{5mm} \overline{\varepsilon}^2(v):=\sup_{v\in\mathcal{V}}E[\varepsilon^2|V=v].
\end{align*}
We denote the best linear approximation to the target function $\theta_0(v)$ by $\theta_k(v):=p(v)'\beta_k$, where
\begin{align*}
    \beta_k:=\underset{b\in\mathbb{R}^k}{\mathrm{argmin}}E[(\theta_0(V)-p(V)'b)^2],
\end{align*}
and the approximation error by $r(v):=\theta_0(v)-\theta_k(v)$.
We use the same notation $\|\cdot\|$ for the $\ell_2$-norm of a vector and the operator norm of a matrix.
The notation $a\lesssim b$ is used when $a\leq cb$ for some positive constant $c$ which does not depend on $N$, and $a\lesssim_P b$ is used when $a=O_P(b)$.
$a\land b$ and $a\lor b$ mean $\min\{a,b\}$ and $\max\{a,b\}$, respectively.
The scaled and demeaned sample average for some function $f$ is denoted by
\begin{align*}
    \mathbb{G}_N[f(v_i)]:=\sqrt{N}E_N[f(v_i)-E[f(V)]].
\end{align*}

\subsection{Pointwise Limit Theory}
The following assumptions are the collection of regularity conditions on the covariates distribution, basis functions, error terms, nuisance estimators, and the target function.
We use these assumptions to establish the pointwise asymptotic theory for the orthogonal estimator.
\begin{assumption}[Identification]\label{as:basis_eigen}
All eigenvalues of $E[p(V)p(V)']$ are bounded above and away from zero uniformly over $k$.
\end{assumption}
\begin{assumption}[Norm of Basis]\label{as:basis_rate}
The sup-norm of the basis functions $\xi_k:=\sup_{v\in\mathcal{V}}\|p(v)\|$ grows sufficiently slow:
\begin{align*}
    \sqrt{\frac{\xi_k^2\log N}{N}}=o(1).
\end{align*}
\end{assumption}
\begin{assumption}[Approximation Error]\label{as:approxerror}
There exists a sequence of finite constants $l_m,r_m$ such that the $L^2$ and sup norms of the approximation error are bounded as follows:
\begin{align*}
    \|r\|_{P,2}:=\sqrt{\int r(v)^2dP(v)}\lesssim r_k,\quad  \|r\|_{P,\infty}:=\sup_{v\in\mathcal{V}}|r(v)|\lesssim l_kr_k.
\end{align*}
\end{assumption}
\begin{assumption}[Stochastic Errors]\label{as:samplingerror}
The second moment of the sampling error conditional on $V$ is bounded from above: $\overline{\varepsilon}^2=O(1)$.
\end{assumption}

\begin{remark}[Plausibility of the Regularity Conditions]
Assumption \ref{as:basis_eigen} through \ref{as:samplingerror} are the regularity conditions widely used in the series estimation literature.
They are plausible enough to hold in many practical situations and for various basis functions.
Indeed, the bounds on $\xi_k$, $r_k$ and $l_k$ have been intensively investigated, and the results show that Assumption \ref{as:basis_rate} and \ref{as:approxerror} are not too restrictive.
\end{remark}

\begin{assumption}[Numerator and Denominator Functions]\label{as:boundedfunc}
The numerator function $\nu_0(v)$ and denominator function $\zeta_0(v)$ are bounded uniformly over $v\in\mathcal{V}$.
Moreover, $\zeta_0(v)\neq0$ for all $v\in\mathcal{V}$.
\end{assumption}
The uniform boundedness of the functions $\nu_0$ and $\zeta_0$ is anyway included in the low-level conditions stated in Section \ref{sec:app_dml}.
Assumption \ref{as:boundedfunc} therefore imposes virtually no extra restriction on these functions.

\begin{assumption}[Small Bias in Nuisance Estimators]\label{as:smallbias}
For all $g\in[G]$, the nuisance estimate $\hat{\eta}_g$, obtained by cross-fitting, belongs to a shrinking neighborhood of $\eta_0$, denoted by $\mathcal{S}_N$.
Uniformly over $\mathcal{S}_N$, the following mean square convergence holds:
\begin{align*}
    B_N&:=\sqrt{N}\left(\sup_{\eta\in\mathcal{S}_N}\|E[p(V)(U-U_0)]\|\lor\sup_{\eta\in\mathcal{S}_N}\|E[p(V)(T-T_0)]\|\right)=o(1),\\
    \Lambda_N&:=\sup_{\eta\in\mathcal{S}_N}E\left[\|p(V)(U-U_0)\|^2\right]^{1/2}\lor\sup_{\eta\in\mathcal{S}_N}E\left[\|p(V)(T-T_0)\|^2\right]^{1/2}=o(1).
\end{align*}
\end{assumption}

\begin{remark}[Sufficient Conditions for the Small Bias]
The low-level sufficient conditions to satisfy Assumption \ref{as:smallbias} in the specific examples are presented in Section \ref{sec:app_dml}.
We can use, for example, deep neural nets for regression \citep{schmidt2020nonparametric,farrell2021deep,kohler2021on} and for classification \citep{kim2021fast,bos2022convergence}, and random forest for regression \citep{wager2015adaptive,syrgkanis2020estimation} and for classification \citep{gao2022towards,peng2022rates}.
\end{remark}

The following is the result on the pointwise convergence rate and linearization, which extends the results obtained in the standard regression setting \citep{belloni2015some,semenova2021debiased}.
\begin{lemma}[Pointwise Convergence Rate and Linearization]\label{le:prate}
Under Assumption \ref{as:basis_eigen}-\ref{as:boundedfunc}, the following statements hold:

\noindent(a) The $\ell_2$-norm of the estimation error is bounded as:
\begin{align*}
    \|\hat{\beta}-\beta_k\|\lesssim_P\sqrt{\frac{k}{N}}+\left(\sqrt{\frac{k}{N}}l_kr_k\land\frac{\xi_kr_k}{\sqrt{N}}\right),
\end{align*}
which implies the same bound on MSE of the estimate $\hat{\theta}$ against the pseudo-target function $\theta_k$:
\begin{align*}
    E_N\left[\left(\hat{\theta}(v_i)-\theta_k(v_i)\right)^2\right]^{1/2}\lesssim_P\sqrt{\frac{k}{N}}+\left(\sqrt{\frac{k}{N}}l_kr_k\land\frac{\xi_kr_k}{\sqrt{N}}\right).
\end{align*}

\noindent(b) For any $\alpha\in\mathcal{A}^{k-1}:=\{\alpha\in\mathbb{R}^k:\|\alpha\|=1\}$, the estimator $\hat{\beta}$ is approximately linear:
\begin{align*}
    \sqrt{N}\alpha'(\hat{\beta}-\beta_k)=\alpha'Q^{-1}\mathbb{G}_N[p_i(\varepsilon_{Ui}+\theta_{0i}\varepsilon_{Ti})]+R_{N}(\alpha),
\end{align*}
where the remainder term $R_{N}(\alpha)$ is bounded as:
\begin{align*}
    R_{N}(\alpha)\lesssim_P B_N+\Lambda_N+l_kr_k+\sqrt{\frac{\xi_k^2\log N}{N}}\left(1+\left(l_kr_k\sqrt{k}\land\xi_kr_k\right)\right).
\end{align*}
\end{lemma}
Lemma \ref{le:prate} states that the OSR converges to the pseudo-target function at the same rate as in the standard regression setting under mild conditions.
However, the linearization result is slightly different from one obtained in \cite{semenova2021debiased}, where the linearization is possible even when the term $p_ir_i$ is included because $E[pr]=0$.
In the CEFR problems, we have the term $p_i\zeta_{0i}r_i$ instead of $p_ir_i$, and $E[p\zeta_0r]\neq0$ in general.

The following theorem establishes the pointwise normality of the OSR estimator with additional conditions to satisfy Lindeberg's condition for the central limit theorem.
\begin{theorem}[Pointwise Normality of the OSR Estimator]\label{th:pnormal}
Suppose Assumption \ref{as:basis_eigen}-\ref{as:boundedfunc} hold.
In addition, suppose (i) $R_N(\alpha)=o(1)$, (ii) $1\lesssim\underline{\varepsilon}^2$ and (iii) $\sup_{v\in\mathcal{V}}E[\varepsilon^21_{|\varepsilon|>M}|V=v]\rightarrow0$ as $M\rightarrow\infty$.
Then, for any $\alpha\in\mathcal{A}^{k-1}$, OSR estimator is asymptotically normal:
\begin{align*}
    \lim_{N\rightarrow\infty}\sup_{e\in\mathbb{R}}\left|P\left(\frac{\alpha'(\hat{\beta}-\beta_k)}{\sqrt{\alpha'\Omega\alpha/N}}<e\right)-\Phi(e)\right|=0.
\end{align*}
Moreover, for any $v_0=v_{0,N}\in\mathcal{V}$, the estimator $\hat{\theta}(v_0)$ against the pseudo-target value $\theta_k(v_0)$ is asymptotically normal:
\begin{align*}
    \lim_{N\rightarrow\infty}\sup_{e\in\mathbb{R}}\left|P\left(\frac{\hat{\theta}(v_0)-\theta_k(v_0)}{\sigma(v_0)/\sqrt{N}}<e\right)-\Phi(e)\right|=0,
\end{align*}
and if the approximation error is negligible relative to the estimation error, namely $r(v_0)=o(\sigma(v_0)/\sqrt{N})$, then $\hat{\theta}(v_0)$ is also asymptotically normal around the true value $\theta_0(v_0)$:
\begin{align*}
    \lim_{N\rightarrow\infty}\sup_{e\in\mathbb{R}}\left|P\left(\frac{\hat{\theta}(v_0)-\theta_0(v_0)}{\sigma(v_0)/\sqrt{N}}<e\right)-\Phi(e)\right|=0.
\end{align*}
\end{theorem}

\subsection{Uniform Limit Theory}
Stronger conditions than needed for the pointwise results are required to establish the uniform asymptotic theory for the OSR estimator.
The following conditions control the behaviour of the stochastic errors, basis functions and nuisance errors more strictly.

\begin{assumption}[Tail Bounds]\label{as:tailbound}
There exists a constant $m>2$ such that the upper bound of the $m$-th moment of $|\varepsilon_U|$ and $|\varepsilon_T|$ is bounded conditional on $V$:
\begin{align*}
    \sup_{v\in\mathcal{V}}E[|\varepsilon|^m|V=v]\lesssim1.
\end{align*}
\end{assumption}

Denote by $\alpha(v):=p(v)/\|p(v)\|$ the normalized value of the basis $p(v)$.
Define the Lipschitz constant for $\alpha(v)$ as:
\begin{align*}
    \xi_k^L=\sup_{v,\tilde{v}\in\mathcal{V},v\neq\tilde{v}}\frac{\|\alpha(v)-\alpha(\tilde{v})\|}{\|v-\tilde{v}\|}.
\end{align*}
\begin{assumption}[Well-Behaved Basis]\label{as:wellbehaved}
Basis functions are well-behaved, namely
(i) $(\xi_k^L)^{2m/(m-2)}\allowbreak\log N/N\lesssim1$ and
(ii) $\log\xi_k^L\lesssim\log k$
for the same $m$ as in Assumption \ref{as:tailbound}.
\end{assumption}

\begin{assumption}[Condition for Matrix Estimation]\label{as:matrixest}
Uniformly over $\mathcal{S}_N$, the following convergence holds:
\begin{align*}
    \kappa_N^1&:=\sup_{\eta\in\mathcal{S}_N}E\left[\max_{1\leq i\leq N}|U_i-U_{0i}|\right]\lor\sup_{\eta\in\mathcal{S}_N}E\left[\max_{1\leq i\leq N}|T_i-T_{0i}|\right]=o(1),\\
    \kappa_N&:=\sup_{\eta\in\mathcal{S}_N}E\left[\max_{1\leq i\leq N}(U_i-U_{0i})^2\right]^{1/2}\lor\sup_{\eta\in\mathcal{S}_N}E\left[\max_{1\leq i\leq N}(T_i-T_{0i})^2\right]^{1/2}=o(1).
\end{align*}
\end{assumption}

The following lemma is about the uniform convergence rate and linearization of the OSR estimator.
\begin{lemma}[Uniform Rate and Uniform Linearization]\label{le:urate}
Suppose Assumption \ref{as:basis_eigen}-\ref{as:wellbehaved} hold.
Then, the following statements hold.

\noindent(a) The OSR estimator is approximately linear uniformly over $\mathcal{V}$:
\begin{align*}
    \sqrt{N}\alpha(v)'(\hat{\beta}-\beta_k)=\alpha(v)'Q^{-1}\mathbb{G}_N[p_i(\varepsilon_{Ui}+\theta_{0i}\varepsilon_{Ti})]+R_N(\alpha(v)),
\end{align*}
where the remainder term $R_N(\alpha(v))$ obeys
\small\begin{align*}
    \sup_{v\in\mathcal{V}}R_N(\alpha(v))&\lesssim_P B_N+\Lambda_N+l_kr_k\sqrt{\log N}+\sqrt{\frac{\xi_k^2\log N}{N}}\left(N^{1/m}\sqrt{\log N}+\left[l_kr_k\sqrt{k}\land\xi_kr_k\right]\right)=:\overline{R}_N.
\end{align*}\normalsize

\noindent(b) The OSR estimator $\hat{\theta}$ of the pseudo-target function $\theta_k$ converges uniformly over $\mathcal{V}$ at the following rate:
\begin{align*}
    \sup_{v\in\mathcal{V}}\left|\hat{\theta}(v)-\theta_k(v)\right|\lesssim_P\frac{\xi_k}{\sqrt{N}}\left(\sqrt{\log N}+\overline{R}_N\right).
\end{align*}
\end{lemma}
Likewise in Lemma \ref{le:prate}, we cannot include the term $p_i\zeta_{0i}r_i$ in the linearization result.
However, its impact is negligible when the basis is sufficiently rich so that $l_kr_k=o(\sqrt{\log N})$.

The following theorem establishes a strong approximation of the OSR estimator's series process.
\begin{theorem}[Strong Approximation by a Gaussian Process]\label{th:strongaprx}
Suppose Assumption \ref{as:basis_eigen}-\ref{as:wellbehaved} hold with $m\geq3$, and let $\overline{a}_N$ be a sequence of positive numbers such that $\overline{a}_N^{-1}=o(1)$.
In addition, suppose (i) $\overline{R}_N=o(\overline{a}_N^{-1})$,
(ii) $1\lesssim\underline{\sigma}^2$ and
(iii) $m^4\overline{a}_N^6\xi_k^2(1+l_k^3r_k^3)^2\log^2N/N=o(1)$.
Then, for some $\mathcal{N}_k\sim N(0,I_k)$, 
\begin{align*}
    \frac{\hat{\theta}(v)-\theta_k(v)}{\sigma(v)/\sqrt{N}}=_d\frac{p(v)'\Omega^{1/2}}{\sigma(v)}\mathcal{N}_k+o_P(\overline{a}_N^{-1})\ in\ \ \ell^\infty(\mathcal{V}).
\end{align*}
In addition, if $\sqrt{N}\sup_{v\in\mathcal{V}}|r(v)|/\sigma(v)=o(\overline{a}_N^{-1})$,
\begin{align*}
    \frac{\hat{\theta}(v)-\theta_0(v)}{\sigma(v)/\sqrt{N}}=_d\frac{p(v)'\Omega^{1/2}}{\sigma(v)}\mathcal{N}_k+o_P(\overline{a}_N^{-1})\ in\ \ \ell^\infty(\mathcal{V}).
\end{align*}
\end{theorem}

Theorem \ref{th:matest} derives the convergence rate of the covariance matrix estimator $\hat{\Omega}$.
\begin{theorem}[Matrices Estimation]\label{th:matest}
Suppose Assumption \ref{as:basis_eigen}-\ref{as:matrixest} hold.
In addition, suppose (i) $\overline{R}_N\lesssim\sqrt{\log N}$ and
(ii) $(N^{1/m}+l_kr_k)(\sqrt{\xi_k^2\log N/N}+\kappa_N^1)=o(1)$.
Then, the covariance matrix estimator $\hat{\Omega}$ converges at the following rate:
\begin{align*}
    \|\hat{\Omega}-\Omega\|\lesssim_P\left(N^{1/m}+l_kr_k\right)\left(\sqrt{\frac{\xi_k^2\log N}{N}}+\kappa_N^1\right)+\kappa_N^2=:a_N.
\end{align*}
Moreover, the following bound holds:
\begin{align*}
    \sup_{v\in\mathcal{V}}\left|\frac{\hat{\sigma}(v)}{\sigma(v)}-1\right|\lesssim_P\|\hat{\Omega}-\Omega\|\lesssim_P a_N.
\end{align*}
\end{theorem}

Theorem \ref{th:bootstrap} establishes the validity of Gaussian bootstrap.
\begin{theorem}[Validity of Gaussian Bootstrap]\label{th:bootstrap}
Suppose the assumptions of Theorem \ref{th:strongaprx} hold with $\overline{a}_N=\log N$ and the assumptions of Theorem \ref{th:matest} hold with $a_N=O(N^{-c})$ for some $c>0$.
In addition, suppose (i) $1\lesssim\underline{\sigma}^2$ and
(ii) there exists a sequence $\xi_N'$ obeying $1\lesssim\xi_N'\lesssim\|p(v)\|$ uniformly for all $v\in\mathcal{V}$ so that $\|p(v)-p(v')\|/\xi_N'\leq L_N\|v-v'\|$, where $\log L_N\lesssim\log N$.
Let $\mathcal{N}_k^b$ be a bootstrap draw from $N(0,I_k)$ and $P^*$ be a probability conditional on data $\{V_i\}_{i=1}^N$.
Then, the following approximation holds uniformly in $\ell^\infty(\mathcal{V})$:
\begin{align*}
    \frac{p(v)'\hat{\Omega}^{1/2}}{\hat{\sigma}(v)}\mathcal{N}_k^b=_d\frac{p(v)'\Omega^{1/2}}{\sigma(v)}\mathcal{N}_k^b+o_{P^*}(\log^{-1}N).
\end{align*}
\end{theorem}

Theorem \ref{eq:confband} is on the validity of the uniform confidence bands and their width.
\begin{theorem}[Validity of Uniform Confidence Bands]\label{th:confband}
Let Assumption \ref{as:basis_eigen}-\ref{as:matrixest} hold with $m\geq4$.
In addition, suppose (i) $\overline{R}_N\lesssim\log^{-1/2}N$,
(ii) $\xi_k\log^2N/N^{1/2-1/m}=o(1)$,
(iii) $1\lesssim\underline{\sigma}^2$,
(iv) $\sup_{v\in\mathcal{V}}\allowbreak\sqrt{N}|r(v)|/\|p(v)\|=o(\log^{-1/2}N)$, and
(v) $k^4\xi_k^2(1+l_k^3r_k^3)^2\log^5N/N=o(1)$.
Then,
\begin{align*}
    P\left(\sup_{v\in\mathcal{V}}|\tau_N(v)|\leq c_N(1-\delta)\right)=1-\delta+o(1)
\end{align*}
for $\tau_N$ defined in (\ref{eq:tstat}).
As a consequence, the confidence bands defined in (\ref{eq:confband}) satisfy
\begin{align*}
   P(\theta_0(v)\in[\underline{i}_N(v),\overline{i}_N(v)]\ \forall v\in\mathcal{V})=1-\delta+o(1).
\end{align*}
The width of the confidence bands obeys the following rate:
\begin{align*}
    \sup_{v\in\mathcal{V}}\left(2c_N(1-\delta)\hat{\sigma}(v)/\sqrt{N}\right)\lesssim_P\sigma(v)\sqrt{\frac{\log N}{N}}\lesssim\sqrt{\frac{\xi_k^2\log N}{N}}.
\end{align*}
\end{theorem}

\section{Applications}\label{sec:app_dml}
We apply the general theoretical results presented in the previous section to the examples in Section \ref{sec:framework}.

\subsection{Local Average Treatment Effect}
Consider the setting of Example \ref{ex:late}, and recall that the parameter of interest is $\theta_0(v)=E[Y_1-Y_0|D_1>D_0,V=v]$.
Below, we provide the identification assumptions for LATE, and the orthogonal signals in this setting.
\begin{assumption}[Identification Assumptions for LATE]\label{as:lateid}
\!
\begin{enumerate}
    \item (Instrument Unconfoundedness). $Y_1,Y_0,D_1,D_0\independent Z|X$.
    \item (Positivity). $0<P(Z=1|X=x)<1$ for all $x\in\mathcal{X}$.
    \item (Instrument Relevance). $P(D_1|V=v)\neq P(D_0|V=v)$ for all $v\in\mathcal{V}$.
    \item (Monotonicity). $P(D_1\geq D_0)=1$.
    \item (Consistency). $D=ZD_1+(1-Z)D_0$ and $Y=DY_1+(1-D)Y_0$.
\end{enumerate}
\end{assumption}
Assumption \ref{as:lateid}.1 states that $Z$ is randomly assigned within a subpopulation sharing the same level of the covariates.
Assumption \ref{as:lateid}.2 restricts the treatment assignment probability from taking extreme values.
Assumption \ref{as:lateid}.3 states that the treatment assignment has non-zero effects on the actual treatment status.
Assumption \ref{as:lateid}.4 excludes defiers, those who never follow the given treatment assignment, from our analysis.
Assumption \ref{as:lateid}.5 relates the potential variables to their realized counterparts.

Consider the following doubly robust signals:
\begin{align*}
    U(O,\eta_0)&=\mu_0(1,X)-\mu_0(0,X)+\frac{Z(Y-\mu_0(1,X))}{\rho_0(X)}-\frac{(1-Z)(Y-\mu_0(0,X))}{1-\rho_0(X)},\\
    T(O,\eta_0)&=\pi_0(1,X)-\pi_0(0,X)+\frac{Z(D-\pi_0(1,X))}{\rho_0(X)}-\frac{(1-Z)(D-\pi_0(0,X))}{1-\rho_0(X)},
\end{align*}
where $\rho_0(x)=P(Z=1|X=x)=E[Z|X=x]$ is a propensity score.
The following theorem establishes the identification of LATE and the Neyman orthogonality of the signals.

\begin{theorem}\label{th:mclate}
Under Assumption \ref{as:lateid}, the following statements hold.

\noindent(a) $\theta_0(v)=\nu_0(v)/\zeta_0(v)$, where $\nu_0(v)$ and $\zeta_0(v)$ are defined in Example \ref{ex:late}.

\noindent(b) $E[U(O,\eta_0)-\nu_0(V)|V=v]=0$ and $E[T(O,\eta_0)-\zeta_0(V)|V=v]=0$ for all $v\in\mathcal{V}$.

\noindent(c) The moment equations in (b) satisfy the Neyman orthogonality condition:
\begin{align*}
    \partial_s E[U(O,\eta_0+s(\eta-\eta_0))-\nu_0(V)|V=v]|_{s=0}&=0,\\
    \partial_s E[T(O,\eta_0+s(\eta-\eta_0))-\zeta_0(V)|V=v]|_{s=0}&=0.
\end{align*}
\end{theorem}

\begin{remark}[One-Sided Noncompliance]
Some experimental designs do not give individuals with $Z=0$ access to the treatment.
This setting is called \emph{one-sided noncompliance}, and formally expressed as $P(D_0=0|X=x)=1$ for all $x\in\mathcal{X}$ \citep{Frolich2013-sk,Donald2014-ca,Kennedy2020-wz}.
Under Assumption \ref{as:lateid} and one-sided noncompliance, LATE is identified as:
\begin{align*}
    \theta_0(v)=\frac{E[E[Y|Z=1,X]-E[Y|Z=0,X]|V=v]}{E[E[D|Z=1,X]|V=v]}.
\end{align*}
Obviously, $\zeta_0(v)>0$ for $v\in\mathcal{V}$ in this case, and thus CV based on the criterion (\ref{eq:cvcriterion}) is possible.
\end{remark}

Then, we give a set of sufficient conditions the nuisance estimators must satisfy so that the general results in Section \ref{sec:asymtheory} hold.
Given the true nuisance functions $\eta_0=\{\mu_0,\pi_0,\rho_0\}$ and sequences of shrinking neighborhoods $\mathcal{S}_N^\mu$ of $\mu_0$, $\mathcal{S}_N^\pi$ of $\pi_0$, $\mathcal{S}_N^\rho$ of $\rho_0$, define the following rates:
\begin{align*}
    \mathbf{m}_{N,q}&:=\sup_{\mu\in \mathcal{S}_N^\mu}\sup_{z\in\{0,1\}}E[(\mu(z,X)-\mu_0(z,X))^q]^{1/q},\\
    \mathbf{p}_{N,q}&:=\sup_{\pi\in \mathcal{S}_N^\pi}\sup_{z\in\{0,1\}}E[(\pi(z,X)-\pi_0(z,X))^q]^{1/q},\\
    \mathbf{r}_{N,q}&:=\sup_{\rho\in \mathcal{S}_N^\rho}E[(\rho(X)-\rho_0(X))^q]^{1/q},
\end{align*}
where $q\geq2$ or $q=\infty$.

\begin{assumption}[First-Stage Rate for LATE]\label{as:fslate}
Assume that there exists a sequence of numbers $\epsilon_N=o(1)$ and sequences of neighborhoods $\mathcal{S}_N^\mu$ of $\mu_0$, $\mathcal{S}_N^\pi$ of $\pi_0$ and $\mathcal{S}_N^\rho$ of $\rho_0$ such that the first-stage estimate $\{\hat{\mu},\hat{\pi},\hat{\rho}\}$ belongs to the set $\mathcal{S}_N^\mu\times \mathcal{S}_N^\pi\times \mathcal{S}_N^\rho$ with probability at least $1-\epsilon_N$.
Assume that mean square rates $\mathbf{m}_{N,2},\mathbf{p}_{N,2},\mathbf{r}_{N,2}$ decay sufficiently fast:
\begin{align*}
    \xi_k(\mathbf{m}_{N,2}\lor\mathbf{p}_{N,2}\lor\mathbf{r}_{N,2})=o(1),
\end{align*}
and one of two alternative conditions holds.
(i) Bounded basis. There exists $C_p<\infty$ so that $\sup_{v\in\mathcal{V}}\allowbreak\|p(v)\|_\infty\leq C_p$, $\sqrt{kN}\mathbf{m}_{N,2}\mathbf{r}_{N,2}=o(1)$ and $\sqrt{kN}\mathbf{p}_{N,2}\mathbf{r}_{N,2}=o(1)$.
(ii) Unbounded basis. There exist $\omega,\psi\in[1,\infty]$, $1/\omega+1/\psi=1$ so that $\sqrt{kN}\mathbf{m}_{N,2\omega}\mathbf{r}_{N,2\psi}=o(1)$, and there exist $\omega',\psi'\in[1,\infty]$, $1/\omega'+1/\psi'=1$ so that $\sqrt{kN}\mathbf{p}_{N,2\omega'}\mathbf{r}_{N,2\psi'}=o(1)$.
Finally, the functions in $\mathcal{S}_N^\mu$, $\mathcal{S}_N^\pi$ and $\mathcal{S}_N^\rho$ are bounded uniformly over their domain:
\begin{align*}
    \sup_{\mu\in \mathcal{S}_N^\mu}\sup_{z\in\{0,1\}}\sup_{x\in\mathcal{X}}|\mu(z,x)|\lor\sup_{\pi\in \mathcal{S}_N^\pi}\sup_{z\in\{0,1\}}\sup_{x\in\mathcal{X}}|\pi(z,x)|\lor\sup_{\rho\in\mathcal{S}_N^\rho}\sup_{x\in\mathcal{X}}|\rho^{-1}(x)|<\overline{C}<\infty.
\end{align*}
\end{assumption}

\begin{corollary}\label{cr:late}
Suppose that Assumption \ref{as:lateid} and \ref{as:fslate} hold.
Then, the orthogonal signals $U$ and $T$ satisfy Assumption \ref{as:smallbias}, and consequently, Theorem \ref{th:pnormal}-\ref{th:confband} hold for LATE if the other assumptions stated in Section \ref{sec:asymtheory} are also satisfied.
\end{corollary}

\subsection{Ratio-Based Treatment Effects}
Consider the setting of Example \ref{ex:rbte}, and recall that the parameter of interest is $\theta_0(v)=E[Y_1|V=v]/E[Y_0|V=v]$.
Below, we provide the identification assumptions for ratio-based CATE, and the orthogonal signals in this setting.
\begin{assumption}[Identification Assumptions for ratio-based CATE]\label{as:rbcate}
\!
\begin{enumerate}
    \item (Unconfoundedness). $Y_1,Y_0\independent D|X$.
    \item (Positivity). $0<P(D=1|X=x)<1$ for all $x\in\mathcal{X}$.
    \item (Non-Zero Outcome). $E[Y_0|V=v]\neq0$ for all $v\in\mathcal{V}$.
    \item (Consistency). $Y=DY_1+(1-D)Y_0$.
\end{enumerate}
\end{assumption}
Assumption \ref{as:rbcate}.1 suppose there are no unobserved confounders.
Assumption \ref{as:rbcate}.2 restricts the treatment probability from taking extreme values.
Assumption \ref{as:rbcate}.3 is a necessary condition to ensure that ratio-based CATE is well-defined uniformly over $\mathcal{V}$.
Assumption \ref{as:rbcate}.4 relates the potential outcomes to the observed outcome.
Note that when we are interested in the odds ratio or hazard ratio, CV based on the criterion (\ref{eq:cvcriterion}) is possible because the outcome $Y$ is positive in these cases.

Consider the following doubly robust signals:
\begin{align*}
    U(O,\eta_0)&=\mu_0(1,X)+\frac{D(Y-\mu_0(1,X))}{\pi_0(X)},\\
    T(O,\eta_0)&=\mu_0(0,X)+\frac{(1-D)(Y-\mu_0(0,X))}{1-\pi_0(X)},
\end{align*}
where $\pi_0(x)=E[D|X=x]$ is a propensity score.
The following theorem establishes the identification of ratio-based CATE and the Neyman orthogonality of the signals.

\begin{theorem}\label{th:mcrb}
Under Assumption \ref{as:rbcate}, the following statements hold.

\noindent(a) $\theta_0(v)=\nu_0(v)/\zeta_0(v)$, where $\nu_0(v)$ and $\zeta_0(v)$ are defined in Example \ref{ex:rbte}.

\noindent(b) $E[U(O,\eta_0)-\nu_0(V)|V=v]=0$ and $E[T(O,\eta_0)-\zeta_0(V)|V=v]=0$ for all $v\in\mathcal{V}$.

\noindent(c) The moment equations in (b) satisfy the Neyman orthogonality condition:
\begin{align*}
    \partial_s E[U(O,\eta_0+s(\eta-\eta_0))-\nu_0(V)|V=v]|_{s=0}&=0,\\
    \partial_s E[T(O,\eta_0+s(\eta-\eta_0))-\zeta_0(V)|V=v]|_{s=0}&=0.
\end{align*}
\end{theorem}

Then, we give a set of sufficient conditions the nuisance estimators must satisfy so that the general results in Section \ref{sec:asymtheory} hold.
Given the true nuisance functions $\eta_0=\{\mu_0,\pi_0\}$ and sequences of shrinking neighborhoods $\mathcal{S}_N^\mu$ of $\mu_0$, $\mathcal{S}_N^\pi$ of $\pi_0$, define the following rates:
\begin{align*}
    \mathbf{m}_{N,q}&:=\sup_{\mu\in \mathcal{S}_N^\mu}\sup_{d\in\{0,1\}}E[(\mu(d,X)-\mu_0(d,X))^q]^{1/q},\\
    \mathbf{p}_{N,q}&:=\sup_{\pi\in \mathcal{S}_N^\pi}E[(\pi(X)-\pi_0(X))^q]^{1/q},
\end{align*}
where $q\geq2$ or $q=\infty$.

\begin{assumption}[First-Stage Rate for Ratio-Based CATE]\label{as:fsrb}
Assume that there exists a sequence of numbers $\epsilon_N=o(1)$ and sequences of neighborhoods $\mathcal{S}_N^\mu$ of $\mu_0$ and $\mathcal{S}_N^\pi$ of $\pi_0$ such that the first-stage estimate $\{\hat{\mu},\hat{\pi}\}$ belongs to the set $\mathcal{S}_N^\mu\times \mathcal{S}_N^\pi$ with probability at least $1-\epsilon_N$.
Assume that mean square rates $\mathbf{m}_{N,2},\mathbf{p}_{N,2}$ decay sufficiently fast:
\begin{align*}
    \xi_k(\mathbf{m}_{N,2}\lor\mathbf{p}_{N,2})=o(1),
\end{align*}
and one of two alternative conditions holds.
(i) Bounded basis. There exists $C_p<\infty$ so that $\sup_{v\in\mathcal{V}}\allowbreak\|p(v)\|_\infty\leq C_p$, $\sqrt{kN}\mathbf{m}_{N,2}\mathbf{p}_{N,2}=o(1)$.
(ii) Unbounded basis. There exists $\omega,\psi\in[1,\infty]$, $1/\omega+1/\psi=1$ so that $\sqrt{kN}\mathbf{m}_{N,2\omega}\mathbf{p}_{N,2\psi}=o(1)$.
Finally, the functions in $\mathcal{S}_N^\mu$ and $\mathcal{S}_N^\pi$ are bounded uniformly over their domain:
\begin{align*}
    \sup_{\mu\in \mathcal{S}_N^\mu}\sup_{d\in\{0,1\}}\sup_{x\in\mathcal{X}}|\mu(d,x)|\lor\sup_{\pi\in \mathcal{S}_N^\pi}\sup_{x\in\mathcal{X}}|\pi(x)^{-1}|<\overline{C}<\infty.
\end{align*}
\end{assumption}

\begin{corollary}\label{cr:rbcate}
Suppose that Assumption \ref{as:rbcate} and \ref{as:fsrb} hold.
Then, the orthogonal signals $U$ and $T$ satisfy Assumption \ref{as:smallbias}, and consequently, Theorem \ref{th:pnormal}-\ref{th:confband} hold for ratio-based CATE if the other assumptions stated in Section \ref{sec:asymtheory} are also satisfied.
\end{corollary}

\subsection{Instrumented Difference-in-Differences}
Consider the setting of Example \ref{ex:idid}, and recall that the parameter of interest is $\theta_0(v)=E[Y_1-Y_0|V=v]$.
Below, we provide the identification assumptions for CATE in the IDID setting, and the orthogonal signals in this setting.
\begin{assumption}[Identification Assumptions for IDID]\label{as:idid}
\!
\begin{enumerate}
    \item (Consistency). $D=D_{zw}$ if $Z=z$ and $W=w$. $Y=Y_{dw}$ if $D=d$ and $W=w$. 
    \item (Positivity). $0<P(Z=z,W=w|X=x)<1$ for all $z,w\in\{0,1\}$ and $x\in\mathcal{X}$.
    \item (Random Sampling). $Y_{dw},D_{zw}\independent W|Z,X$ for $d,z,w\in\{0,1\}$.
    \item (Trend Relevance). $E[D_{11}-D_{10}|X]\neq E[D_{01}-D_{00}|X]$.
    \item (Trend Unconfoundedness). $D_{zw},Y_{01}-Y_{00},Y_{1w}-Y_{0w}\independent Z|X$ for $z,w\in\{0,1\}$.
    \item (No Unmeasured Common Effect Modifier). $Cov(D_{1w}-D_{0w},Y_{1w}-Y_{0w}|X)=0$ for $w\in\{0,1\}$.
    \item (Stable Treatment Effect Over Time). $E[Y_1-Y_0|X]:=E[Y_{11}-Y_{01}|X]=E[Y_{10}-Y_{00}|X]$.
\end{enumerate}
\end{assumption}
Assumption \ref{as:idid}.1 relates the potential variables to their realized counterparts.
Assumption \ref{as:idid}.2 restricts the propensity score from taking extreme values.
Assumption \ref{as:idid}.3 states that the distributions of the potential variables remain the same over time when conditioned on $Z$ and $X$.
Assumption \ref{as:idid}.4, 5 and 6 are weaker than the usual IV assumptions, where $Z$ is correlated with $D$ but independent of every potential variable.
Instead, IDID only assumes $Z$ is a valid IV for the difference of the potential variables.
Assumption \ref{as:idid} states that the magnitude of a treatment effect does not change over time.

Consider the following signals
\begin{align*}
    U(O,\eta_0)&=\sum_{(w,z)\in\{0,1\}^2}(-1)^{w+z}\left(\mu_0(w,z,X)+\frac{1_{W=w}1_{Z=z}(Y-\mu_0(w,z,X))}{\rho_0(w,z,X)}\right),\\
    T(O,\eta_0)&=\sum_{(w,z)\in\{0,1\}^2}(-1)^{w+z}\left(\pi_0(w,z,X)+\frac{1_{W=w}1_{Z=z}(D-\pi_0(w,z,X))}{\rho_0(w,z,X)}\right),
\end{align*}
where $\rho_0(w,z,x)=P(W=w,Z=z|X=x)$ and $1_A$ is an indicator function that takes 1 if $A$ is true.
Note that estimating $\rho$ can be done by constructing a four-class classifier and then obtaining posterior probabilities.
The following theorem establishes the identification of CATE in IDID and the Neyman orthogonality of the signals.

\begin{theorem}\label{th:mcidid}
Under Assumption \ref{as:idid}, the following statements hold.

\noindent(a) $\theta_0(v)=\nu_0(v)/\zeta_0(v)$, where $\nu_0(v)$ and $\zeta_0(v)$ are defined in Example \ref{ex:idid}.

\noindent(b) $E[U(O,\eta_0)-\nu_0(V)|V=v]=0$ and $E[T(O,\eta_0)-\zeta_0(V)|V=v]=0$ for all $v\in\mathcal{V}$.

\noindent(c) The moment equations in (b) satisfy the Neyman orthogonality condition:
\begin{align*}
    \partial_s E[U(O,\eta_0+s(\eta-\eta_0))-\nu_0(V)|V=v]|_{s=0}&=0,\\
    \partial_s E[T(O,\eta_0+s(\eta-\eta_0))-\zeta_0(V)|V=v]|_{s=0}&=0.
\end{align*}
\end{theorem}

Then, we give a set of sufficient conditions the nuisance estimators must satisfy so that the general results in Section \ref{sec:asymtheory} hold.
Given the true nuisance functions $\eta_0=\{\mu_0,\pi_0,\rho_0\}$ and sequences of shrinking neighborhoods $\mathcal{S}_N^\mu$ of $\mu_0$, $\mathcal{S}_N^\pi$ of $\pi_0$, $\mathcal{S}_N^\rho$ of $\rho_0$, define the following rates:
\begin{align*}
    \mathbf{m}_{N,q}&:=\sup_{\mu\in \mathcal{S}_N^\mu}\sup_{(w,z)\in\{0,1\}^2}E[(\mu(w,z,X)-\mu_0(w,z,X))^q]^{1/q},\\
    \mathbf{p}_{N,q}&:=\sup_{\pi\in \mathcal{S}_N^\pi}\sup_{(w,z)\in\{0,1\}^2}E[(\pi(w,z,X)-\pi_0(w,z,X))^q]^{1/q},\\
    \mathbf{r}_{N,q}&:=\sup_{\rho\in \mathcal{S}_N^\rho}\sup_{(w,z)\in\{0,1\}^2}E[(\rho(w,z,X)-\rho_0(w,z,X))^q]^{1/q},
\end{align*}
where $q\geq2$ or $q=\infty$.

\begin{assumption}[First-Stage Rate for IDID]\label{as:fsidid}
Assume that there exists a sequence of numbers $\epsilon_N=o(1)$ and sequences of neighborhoods $\mathcal{S}_N^\mu$ of $\mu_0$, $\mathcal{S}_N^\pi$ of $\pi_0$ and $\mathcal{S}_N^\rho$ of $\rho_0$ such that the first-stage estimate $\{\hat{\mu},\hat{\pi},\hat{\rho}\}$ belongs to the set $\mathcal{S}_N^\mu\times \mathcal{S}_N^\pi\times \mathcal{S}_N^\rho$ with probability at least $1-\epsilon_N$.
Assume that mean square rates $\mathbf{m}_{N,2},\mathbf{p}_{N,2},\mathbf{r}_{N,2}$ decay sufficiently fast:
\begin{align*}
    \xi_k(\mathbf{m}_{N,2}\lor\mathbf{p}_{N,2}\lor\mathbf{r}_{N,2})=o(1),
\end{align*}
and one of two alternative conditions holds.
(i) Bounded basis. There exists $C_p<\infty$ so that $\sup_{v\in\mathcal{V}}\allowbreak\|p(v)\|_\infty\leq C_p$, $\sqrt{kN}\mathbf{m}_{N,2}\mathbf{r}_{N,2}=o(1)$ and $\sqrt{kN}\mathbf{p}_{N,2}\mathbf{r}_{N,2}=o(1)$.
(ii) Unbounded basis. There exist $\omega,\psi\in[1,\infty]$, $1/\omega+1/\psi=1$ so that $\sqrt{kN}\mathbf{m}_{N,2\omega}\mathbf{r}_{N,2\psi}=o(1)$, and there exist $\omega',\psi'\in[1,\infty]$, $1/\omega'+1/\psi'=1$ so that $\sqrt{kN}\mathbf{p}_{N,2\omega'}\mathbf{r}_{N,2\psi'}=o(1)$.
Finally, the functions in $\mathcal{S}_N^\mu$, $\mathcal{S}_N^\pi$ and $\mathcal{S}_N^\rho$ are bounded uniformly over their domain:
\begin{align*}
    \sup_{\mu\in \mathcal{S}_N^\mu}\sup_{(w,z)\in\{0,1\}^2}\sup_{x\in\mathcal{X}}|\mu(w,z,x)|&\lor\sup_{\pi\in \mathcal{S}_N^\pi}\sup_{(w,z)\in\{0,1\}^2}\sup_{x\in\mathcal{X}}|\pi(w,z,x)|\\
    &\qquad\lor\sup_{\rho\in \mathcal{S}_N^\rho}\sup_{(w,z)\in\{0,1\}^2}\sup_{x\in\mathcal{X}}|\rho^{-1}(w,z,x)|<\overline{C}<\infty.
\end{align*}
\end{assumption}

\begin{corollary}\label{cr:idid}
Suppose that Assumption \ref{as:idid} and \ref{as:fsidid} hold.
Then, the orthogonal signals $U$ and $T$ satisfy Assumption \ref{as:smallbias}, and consequently, Theorem \ref{th:pnormal}-\ref{th:confband} hold for CATE in IDID if the other assumptions stated in Section \ref{sec:asymtheory} are also satisfied.
\end{corollary}

\subsection{Treatment Effects in Data Combination Setting}
Consider the setting of Example \ref{ex:dcte}, and recall that the parameter of interest is CATE $\theta_0(v)=E[Y_1-Y_0|V=v]$ or LATE $\theta_0(v)=E[Y_1-Y_0|D_1>D_0,V=v]$.
Below, we provide the identification assumptions for CATE and LATE in the data combination setting, and the orthogonal signals in this setting.

\begin{assumption}[Identification Assumptions for the Data Combination Setting]\label{as:dcid}
\!
\begin{enumerate}
    \item (Random Sampling). $Y_1,Y_0,D_1,D_0,Z_1,Z_0\independent H,W|X$.
    \item (Instrument Unconfoundedness). $Y_1,Y_0,D_1,D_0\independent Z_1,Z_0|X$.
    \item (Positivity). $0<P(H=h,W=w|X=x)<1$ for $h,w=0,1$ and all $x\in\mathcal{X}$.
    \item (Different Treatment Regimes and Instrument Relevance). $P(Z_1|V=v)\neq P(Z_0|V=v)$ and $P(D_1|V=v)\neq P(D_0|V=v)$ for all $v\in\mathcal{V}$.
    \item (Consistency). $Z=WZ_1+(1-W)Z_0$, $D=ZD_1+(1-Z)D_0$ and $Y=DY_1+(1-D)Y_0$.
\end{enumerate}
\end{assumption}
Assumption \ref{as:dcid}.1 states that $H$ and $W$ are randomly assigned within the subpopulation sharing the same level of covariates.
Likewise, Assumption \ref{as:dcid}.2 states that $Z$ is randomly assigned within strata defined by $X$.
Assumption \ref{as:dcid}.3 can be automatically satisfied as long as we have four different datasets.
Assumption \ref{as:dcid}.4 is the key assumption in this setting, which states that the two treatment regimes are different, and the treatment assignment has some impact on the individual's decision to receive a treatment or not.
Assumption \ref{as:dcid}.5 relates the potential variables to their realized counterparts.
We need $P(D_1\geq D_0|V=v)=1$ for $v\in\mathcal{V}$ when the parameter of interest is LATE, and $P(D_1>D_0|V=v)=1$ for $v\in\mathcal{V}$ when the parameter of interest is CATE.
\begin{remark}[Facilitating Data Collection]\label{rm:datacollection}
As pointed out in \cite{shinoda2022estimation}, we do not need a positivity condition on $Z_1$ and $Z_0$ as long as Assumption \ref{as:dcid}.4 is satisfied.
We can use a dataset with $P(Z_0=1|V=v)=0$ for all $v\in\mathcal{V}$ if $P(Z_1=1|V=v)\neq0$ for all $v\in\mathcal{V}$, which suggests the use of datasets collected from individuals without any intervention.
Thus, although we need two different treatment regimes in this setting, a single intervention is sufficient in practice.
Datasets without any intervention are available, for example, as government statistics at no cost.
Moreover, we can regard this problem setting as the repeated cross-sectional design, where $W=0,1$ represents the time point before and after the intervention is carried out, respectively.
In this sense, the structure of this setting is similar to the difference-in-differences.
Using a dataset with no intervention also benefits the model selection in CATE estimation and LATE estimation under one-sided noncompliance since CV based on (\ref{eq:cvcriterion}) becomes valid.
We have that
\begin{align*}
    \zeta_0(v)=E[Z_1D_1|V=v]>0,
\end{align*}
by $P(Z_0=1|V=v)=0$, $P(D_0=0|V=v)=1$ and Assumption \ref{as:dcid}.
\end{remark}

Consider the following signals
\begin{align*}
    U(O,\eta_0)&:=\mu_0(1,X)-\mu_0(0,X)+\frac{HW(Y-\mu_0(1,X))}{\rho_0(1,1,X)}-\frac{H(1-W)(Y-\mu_0(0,X))}{\rho_0(1,0,X)},\\
    T(O,\eta_0)&:=\pi_0(1,X)-\pi_0(0,X)+\frac{(1-H)W(D-\pi_0(1,X))}{\rho_0(0,1,X)}-\frac{(1-H)(1-W)(D-\pi_0(0,X))}{\rho_0(0,0,X)},
\end{align*}
where $\rho_0(h,w,x)=P(H=h,W=w|X=x)$.
The same comment for $\rho_0$ in Example \ref{ex:idid} also applies here.
The following theorem establishes the identification of CATE and LATE in the data combination setting and the Neyman orthogonality of the signals.

\begin{theorem}\label{th:mcdc}
Suppose Assumption \ref{as:dcid} holds.
Then, the following statements hold for LATE if $P(D_1\geq D_0|V=v)=1$ holds for all $v\in\mathcal{V}$.
The same statements hold for CATE if $P(D_1>D_0|V=v)=1$ holds for all $v\in\mathcal{V}$.

\noindent(a) $\theta_0(v)=\nu_0(v)/\zeta_0(v)$, where $\nu_0(v)$ and $\zeta_0(v)$ are defined in Example \ref{ex:dcte}.

\noindent(b) $E[U(O,\eta_0)-\nu_0(V)|V=v]=0$ and $E[T(O,\eta_0)-\zeta_0(V)|V=v]=0$ for all $v\in\mathcal{V}$.

\noindent(c) The moment equations in (b) satisfy the Neyman orthogonality condition:
\begin{align*}
    \partial_s E[U(O,\eta_0+s(\eta-\eta_0))-\nu_0(V)|V=v]|_{s=0}&=0,\\
    \partial_s E[T(O,\eta_0+s(\eta-\eta_0))-\zeta_0(V)|V=v]|_{s=0}&=0.
\end{align*}
\end{theorem}

Then, we give a set of sufficient conditions the nuisance estimators must satisfy so that the general results in Section \ref{sec:asymtheory} hold.
Given the true nuisance functions $\eta_0=\{\mu_0,\pi_0,\rho_0\}$ and sequences of shrinking neighborhoods $\mathcal{S}_N^\mu$ of $\mu_0$, $\mathcal{S}_N^\pi$ of $\pi_0$, $\mathcal{S}_N^\rho$ of $\rho_0$, define the following rates:
\begin{align*}
    \mathbf{m}_{N,q}&:=\sup_{\mu\in \mathcal{S}_N^\mu}\sup_{w\in\{0,1\}}E[(\mu(w,X)-\mu_0(w,X))^q]^{1/q},\\
    \mathbf{p}_{N,q}&:=\sup_{\pi\in \mathcal{S}_N^\pi}\sup_{w\in\{0,1\}}E[(\pi(w,X)-\pi_0(w,X))^q]^{1/q},\\
    \mathbf{r}_{N,q}&:=\sup_{\rho\in \mathcal{S}_N^\rho}\sup_{(h,w)\in\{0,1\}^2}E[(\rho(h,w,X)-\rho_0(h,w,X))^q]^{1/q},
\end{align*}
where $q\geq2$ or $q=\infty$.

\begin{assumption}[First-Stage Rate for Data Combination Settings]\label{as:fsdc}
Assume that there exists a sequence of numbers $\epsilon_N=o(1)$ and sequences of neighborhoods $\mathcal{S}_N^\mu$ of $\mu_0$, $\mathcal{S}_N^\pi$ of $\pi_0$ and $\mathcal{S}_N^\rho$ of $\rho_0$ such that the first-stage estimate $\{\hat{\mu},\hat{\pi},\hat{\rho}\}$ belongs to the set $\mathcal{S}_N^\mu\times \mathcal{S}_N^\pi\times \mathcal{S}_N^\rho$ with probability at least $1-\epsilon_N$.
Assume that mean square rates $\mathbf{m}_{N,2},\mathbf{p}_{N,2},\mathbf{r}_{N,2}$ decay sufficiently fast:
\begin{align*}
    \xi_k(\mathbf{m}_{N,2}\lor\mathbf{p}_{N,2}\lor\mathbf{r}_{N,2})=o(1),
\end{align*}
and one of two alternative conditions holds.
(i) Bounded basis. There exists $C_p<\infty$ so that $\sup_{v\in\mathcal{V}}\allowbreak\|p(v)\|_\infty\leq C_p$, $\sqrt{kN}\mathbf{m}_{N,2}\mathbf{r}_{N,2}=o(1)$ and $\sqrt{kN}\mathbf{p}_{N,2}\mathbf{r}_{N,2}=o(1)$.
(ii) Unbounded basis. There exist $\omega,\psi\in[1,\infty]$, $1/\omega+1/\psi=1$ so that $\sqrt{kN}\mathbf{m}_{N,2\omega}\mathbf{r}_{N,2\psi}=o(1)$, and there exist $\omega',\psi'\in[1,\infty]$, $1/\omega'+1/\psi'=1$ so that $\sqrt{kN}\mathbf{p}_{N,2\omega'}\mathbf{r}_{N,2\psi'}=o(1)$.
Finally, the functions in $\mathcal{S}_N^\mu$, $\mathcal{S}_N^\pi$ and $\mathcal{S}_N^\rho$ are bounded uniformly over their domain:
\begin{align*}
    \sup_{\mu\in \mathcal{S}_N^\mu}\sup_{w\in\{0,1\}}\sup_{x\in\mathcal{X}}|\mu(w,x)|&\lor\sup_{\pi\in \mathcal{S}_N^\pi}\sup_{w\in\{0,1\}}\sup_{x\in\mathcal{X}}|\pi(w,x)|\\
    &\qquad\lor\sup_{\rho\in \mathcal{S}_N^\rho}\sup_{(h,w)\in\{0,1\}^2}\sup_{x\in\mathcal{X}}|\rho^{-1}(h,w,x)|<\overline{C}<\infty.
\end{align*}
\end{assumption}

\begin{corollary}\label{cr:dc}
Suppose that Assumption \ref{as:dcid} and \ref{as:fsdc} hold in addition to the assumptions in Theorem \ref{th:mcdc}.
Then, the orthogonal signals $U$ and $T$ satisfy Assumption \ref{as:smallbias}, and consequently, Theorem \ref{th:pnormal}-\ref{th:confband} hold for CATE and LATE in the data combination setting if the other assumptions stated in Section \ref{sec:asymtheory} are also satisfied.
\end{corollary}

\section{Simulations}\label{sec:simdml}
In this section, we conduct numerical simulations to evaluate the finite-sample performance of the proposed method.
We first compare DSR developed in Section \ref{sec:direct} to the existing methods, and then illustrate the performance of OSR and its uniform confidence band.

\subsection{Direct Series Ratio Estimator}\label{sec:dsrsim}
\textbf{Setting.}
In this simulation study, we focus on CATE estimation by data combination explained in Example \ref{ex:dcte}.
We generate random samples of $O=(HY+(1-H)D,W,H,X)$ with a sample size $N=500,1000,2000$ using the following two data generating processes (DGPs):
\begin{align*}
    \mathbf{DGP-L}\hspace{5mm}Y_D&=\varsigma(X_1+X_2)+0.4D(X_1+X_2)+0.2\epsilon,\\
    P(D_1=1|X)&=\varsigma(X_1+X_2),\hspace{5mm}P(D_0=1|X)=0,\hspace{5mm}P(W=1|X)=P(H=1|X)=0.5,\\
    \mathbf{DGP-Q}\hspace{5mm}Y_D&=\varsigma(X_1+X_2)+0.2D(X_1+X_2)^2+0.2\epsilon,\\
    P(D_1=1|X)&=\varsigma(X_1+X_2),\hspace{5mm}P(D_0=1|X)=0,\hspace{5mm}P(W=1|X)=P(H=1|X)=0.5,
\end{align*}
where $X_1,X_2,\epsilon$ are mutually independent random variables drawn from $N(0,1)$, and $\varsigma(x)=(1+e^{-x})^{-1}$ is the logistic function.
Note that values of $H$ and $W$ are determined completely at random independent of $X$ so that DSR is applicable, which is the original setting considered in \cite{yamane}.
Also, we consider the case where no intervention takes place in regime 0 ($P(D_0=1|X)=0$) since it enables the CV method in Remark \ref{rm:cv} by ensuring the denominator is strictly positive.
See Remark \ref{rm:datacollection} for this point.
The number of replication is set $1000$.

We consider the following candidate basis vectors:
\begin{gather*}
    p_3(X)=(1,X_1,X_2)',\hspace{5mm}p_6(X)=(p_3(X)',X_1^2,X_2^2,X_1X_2)',\\
    p_{10}(X)=(p_6(X)',X_1^3,X_2^3,X_1^2X_2,X_1X_2^2)',
\end{gather*}
and regularization parameters $\lambda\in\{0.001,0.01,0.1,1\}$.
These hyperparameters are selected based on 5-fold CV in each replication using the criterion (\ref{eq:cvcriterion}).
For comparison with DSR, we consider a naive separate estimation (SEP), the Direct Least Squares (DLS) proposed in \cite{yamane} and the Directly Weighted Least Squares (DWLS) proposed in \cite{shinoda2022estimation}.
The estimation of $\pi_0$, when necessary, is implemented by the $\ell_2$-regularized logistic regression with a regularization parameter $\lambda$.

\begin{table*}[t]
\centering
\caption{Selected Hyperparameters.}
\begin{tabular}{cccccccccccccccccc}
\hline
     & \multicolumn{8}{c}{DGP-L}                                                                       &  & \multicolumn{8}{c}{DGP-Q}                                                                       \\ \cline{2-9} \cline{11-18} 
     & \multicolumn{2}{c}{$N=500$} &  & \multicolumn{2}{c}{$N=1000$} &  & \multicolumn{2}{c}{$N=2000$} &  & \multicolumn{2}{c}{$N=500$} &  & \multicolumn{2}{c}{$N=1000$} &  & \multicolumn{2}{c}{$N=2000$} \\ \cline{2-3} \cline{5-6} \cline{8-9} \cline{11-12} \cline{14-15} \cline{17-18} 
     & $k$       & $\lambda$       &  & $k$        & $\lambda$       &  & $k$        & $\lambda$       &  & $k$       & $\lambda$       &  & $k$        & $\lambda$       &  & $k$        & $\lambda$       \\ \hline
DSR  & 3         & 0.1             &  & 3          & 0.1             &  & 3          & 0.01            &  & 3         & 0.1             &  & 6          & 0.1             &  & 6          & 0.1             \\
SEP1 & 6         & 0.1             &  & 6          & 0.01            &  & 6          & 0.1             &  & 6         & 0.1             &  & 6          & 0.1             &  & 6          & 0.1             \\
SEP2 & 6         & 1               &  & 6          & 0.01            &  & 6          & 1               &  & 6         & 1               &  & 6          & 0.01            &  & 6          & 1               \\
DLS1 & 3         & 0.01            &  & 3          & 0.01            &  & 3          & 0.01            &  & 10        & 0.01            &  & 6          & 0.1             &  & 6          & 0.1             \\
DLS2 & 10        & 0.001           &  & 10         & 0.001           &  & 10         & 0.001           &  & 10        & 0.001           &  & 10         & 0.001           &  & 10         & 0.01            \\
DWLS & 3         & 0.1             &  & 3          & 0.1             &  & 3          & 0.1             &  & 3         & 0.1             &  & 3          & 0.1             &  & 3          & 0.1             \\ \hline
\end{tabular}
\label{tb:selected_params}
\end{table*}

\begin{table*}[t]
\centering
\caption{Simulation Results for DGP-L.}
\begin{tabular}{cccccccccccc}
\hline
     & \multicolumn{3}{c}{$N=500$} &  & \multicolumn{3}{c}{$N=1000$} &  & \multicolumn{3}{c}{$N=2000$} \\ \cline{2-4} \cline{6-8} \cline{10-12} 
     & Bias    & SD     & MSE    &  & Bias     & SD     & MSE    &  & Bias     & SD     & MSE    \\ \hline
DSR & -0.116  & 0.321  & 0.135  &  & -0.075   & 0.232  & 0.058  &  & -0.044   & 0.164  & 0.029  \\
SEP  & 0.489   & 0.537  & 1.273  &  & 0.579    & 0.374  & 0.877  &  & 0.612    & 0.256  & 0.547  \\
DLS  & -0.239  & 0.363  & 0.200  &  & -0.168   & 0.265  & 0.124  &  & -0.120   & 0.172  & 0.062  \\
DWLS & -0.204  & 0.343  & 0.139  &  & -0.175   & 0.224  & 0.103  &  & -0.181   & 0.176  & 0.099  \\ \hline
\end{tabular}
\label{tb:res_dgpl}
\end{table*}
\begin{table*}[t]
\centering
\caption{Simulation Results for DGP-Q.}
\begin{tabular}{cccccccccccc}
\hline
     & \multicolumn{3}{c}{$N=500$} &  & \multicolumn{3}{c}{$N=1000$} &  & \multicolumn{3}{c}{$N=2000$} \\ \cline{2-4} \cline{6-8} \cline{10-12} 
     & Bias    & SD     & MSE    &  & Bias     & SD     & MSE    &  & Bias     & SD     & MSE    \\ \hline
DSR & -0.155  & 0.403  & 0.426  &  & -0.108   & 0.289  & 0.208  &  & -0.069   & 0.196  & 0.105  \\
SEP  & 0.642   & 0.685  & 2.011  &  & 0.772    & 0.492  & 1.372  &  & 0.857    & 0.330  & 1.019  \\
DLS  & -0.243  & 0.460  & 0.536  &  & -0.179   & 0.319  & 0.377  &  & -0.167   & 0.214  & 0.249  \\
DWLS & -0.110  & 0.446  & 0.366  &  & -0.072   & 0.333  & 0.282  &  & -0.060   & 0.287  & 0.224  \\ \hline
\end{tabular}
\label{tb:res_dgpq}
\end{table*}

\begin{table}[t]
\centering
\caption{Computation Time for One Estimation (Standardized by DSR).}
\begin{tabular}{cccccc}
\hline
     & $N=500$   &  & $N=1000$  &  & $N=2000$  \\ \hline
DSR & 1.000 &  & 1.000 &  & 1.000 \\
SEP  & 5.614 &  & 6.350 &  & 9.044 \\
DLS  & 1.704 &  & 1.616 &  & 1.747 \\
DWLS & 2.198 &  & 2.155 &  & 2.166 \\ \hline
\end{tabular}
\label{tb:computation_time}
\end{table}

\noindent\textbf{Results.}
Table \ref{tb:selected_params} summarises the hyperparameters most frequently selected by 5-fold CV.
SEP1 and SEP2 in the table represent the estimation of the numerator and denominator for the SEP estimator, respectively, while DLS1 and DLS2 represent the estimation of CATE and the auxiliary function, respectively.
It can be seen that CV for DSR is working well as a smaller $\lambda$ is chosen as the sample size grows in DGP-L, and the optimal basis is selected in DGP-Q.
However, for other estimators, CV sometimes fails to select the small regularization and optimal basis even with the large sample size.

Table \ref{tb:res_dgpl} and \ref{tb:res_dgpq} are the simulation results for DGP-L and DGP-Q, respectively.
Bias and SD in the tables represent the bias and standard deviation of the estimates at $(X_1,X_2)=(1,1)$, respectively, and MSE is calculated on the test samples.
In most cases, DSR outperforms the other estimators, while it has a slightly larger bias than DWLS in DGP-Q.
However, DSR has a significantly smaller MSE than DWLS as the sample size grows, which implies that the former is more efficient.
DSR is also found to have an advantage over other estimators in terms of computational efficiency.
As shown in Table \ref{tb:computation_time}, DSR requires the shortest computation time for estimation (without CV), reflecting its simple estimation procedures.
Moreover, the computational advantage of DSR becomes even clearer by taking CV into account.
For example, since there are 3 possible basis vectors and 4 possible regularization parameters for each function in this simulation, we must repeat CV 24 times for SEP and DWLS and 144 times for DLS, while DSR requires only 12 times as it has no nuisance estimation.

Then, we conduct the sensitivity analysis to evaluate the impact of the hyperparameters $k,\lambda$ on the performance of the estimators.
The detailed results for all estimators are summarised in Table \ref{tb:sens_dgpl} and \ref{tb:sens_dgpq} in Section \ref{sec:sensitivity} of Appendix.
Figure \ref{fg:sens_dgpl} and \ref{fg:sens_dgpq} illustrate the change in MSE of DSR with different hyperparameter combinations.
It can be seen that smaller $\lambda$ gives smaller MSE as the number of samples increases in both DGPs.
However, in DGP-L, MSE changes significantly with the choice of $\lambda$ as the sample size grows, whereas in DGP-Q, MSE does not respond to $\lambda$ to the same extent at a large sample size.
Therefore, when the number of samples is more than 1000 in DGP-L, MSE can be small depending on the choice of $\lambda$ even if the model is not optimal ($k=3$), but in DGP-Q, if the model is too simple ($k=3$) or too complex ($k=10$), MSE is considerably larger and the importance of the choice of $k$ is relatively high.
Also, by comparing the results of the sensitivity analysis and Table \ref{tb:selected_params}, we can see that the CV method for DSR successfully selected the optimal hyperparameters in all cases except when $N=500$ in DGP-Q.

\begin{figure*}[!t]
\centering
  \begin{minipage}[b]{0.3\linewidth}
    \centering
    \includegraphics[width=0.9\linewidth]{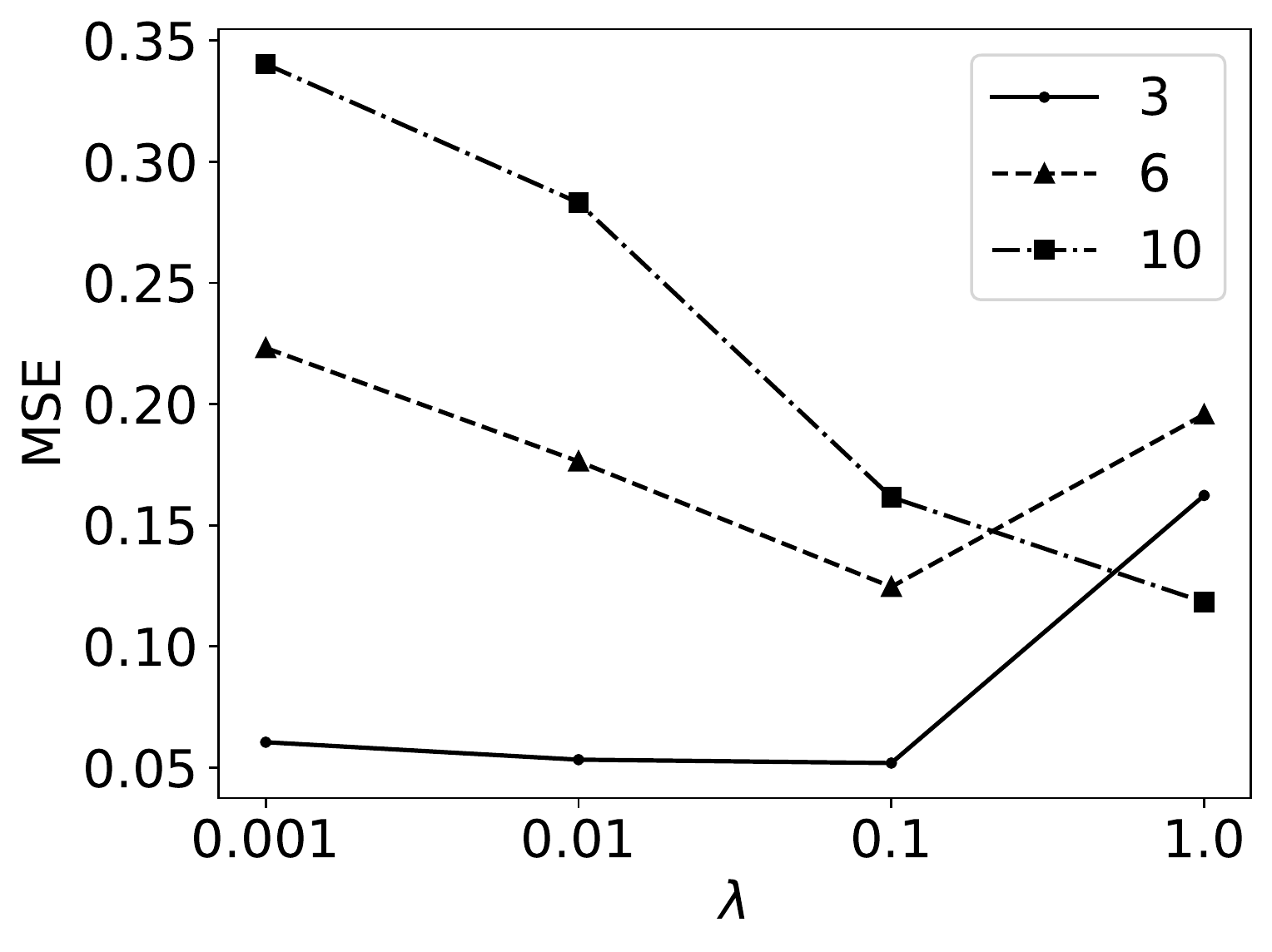}
    \caption*{(a) $N=500$.}
  \end{minipage}
  \begin{minipage}[b]{0.3\linewidth}
    \centering
    \includegraphics[width=0.9\columnwidth]{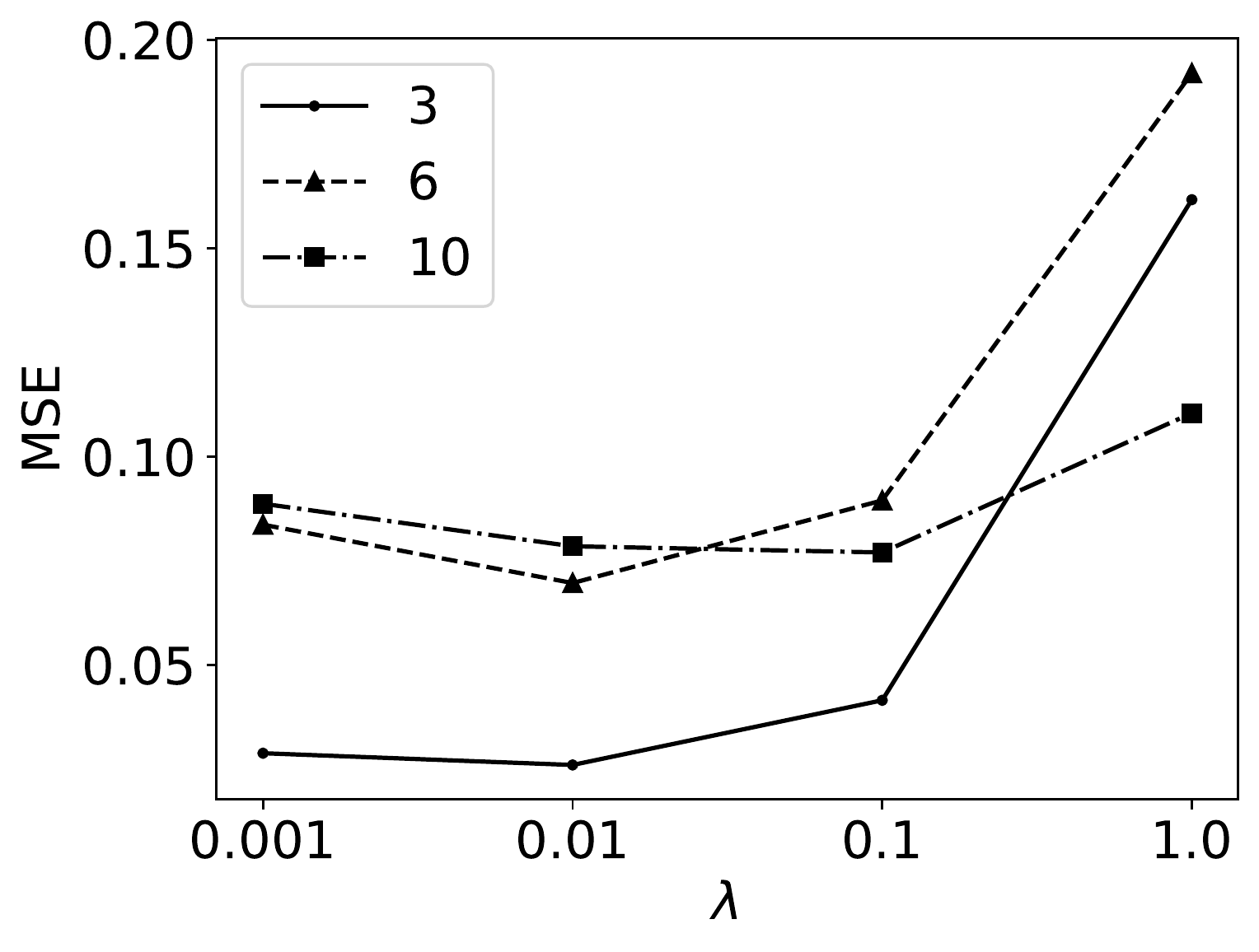}
    \caption*{(b) $N=1000$.}
  \end{minipage}
  \begin{minipage}[b]{0.3\linewidth}
    \centering
    \includegraphics[width=0.9\columnwidth]{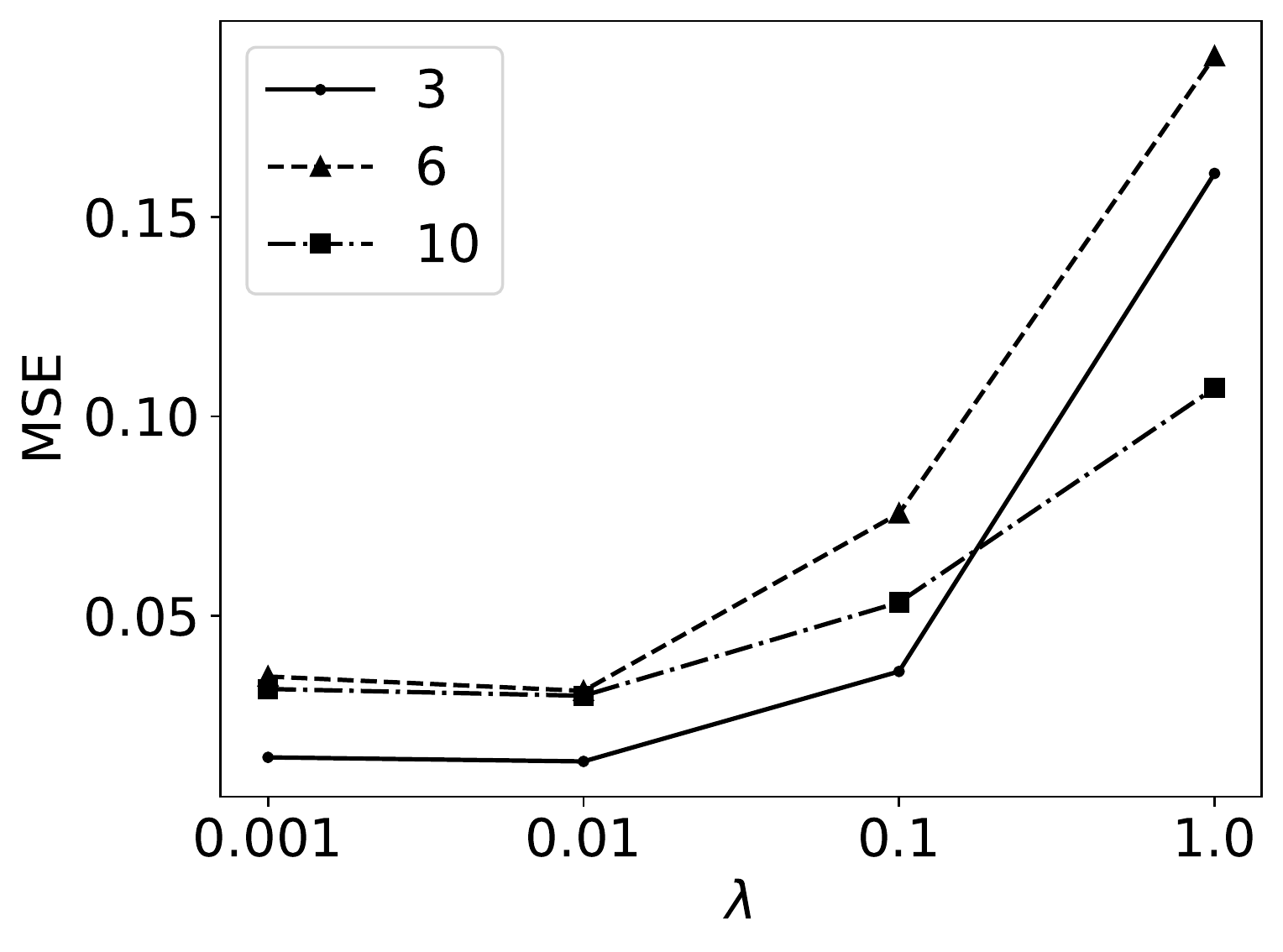}
    \caption*{(c) $N=2000$.}
  \end{minipage}
  \caption{Sensitivity Analysis for DGP-L.}
  \label{fg:sens_dgpl}
\end{figure*}
\begin{figure*}[!t]
\centering
  \begin{minipage}[b]{0.3\linewidth}
    \centering
    \includegraphics[width=0.9\columnwidth]{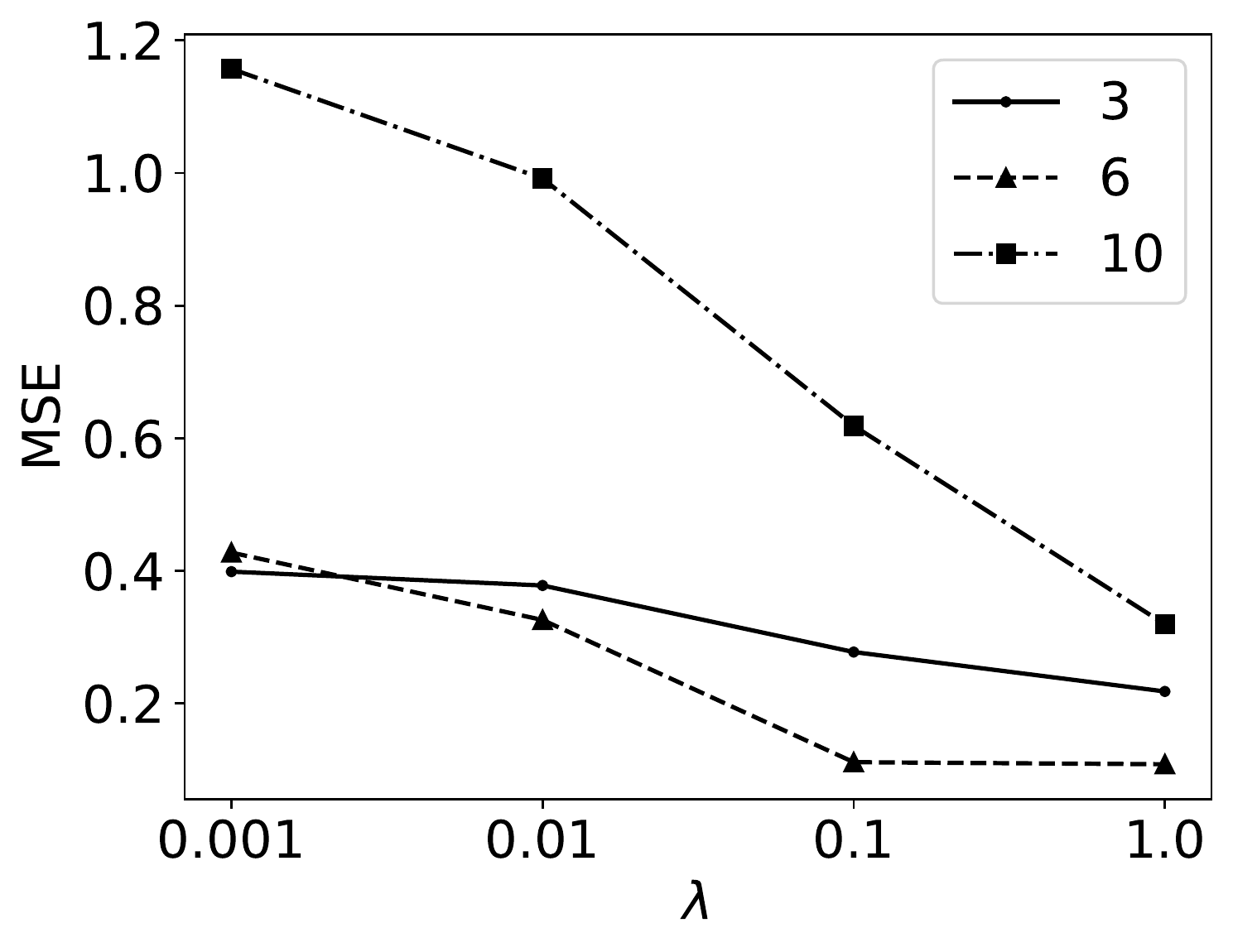}
    \caption*{(a) $N=500$.}
  \end{minipage}
  \begin{minipage}[b]{0.3\linewidth}
    \centering
    \includegraphics[width=0.9\columnwidth]{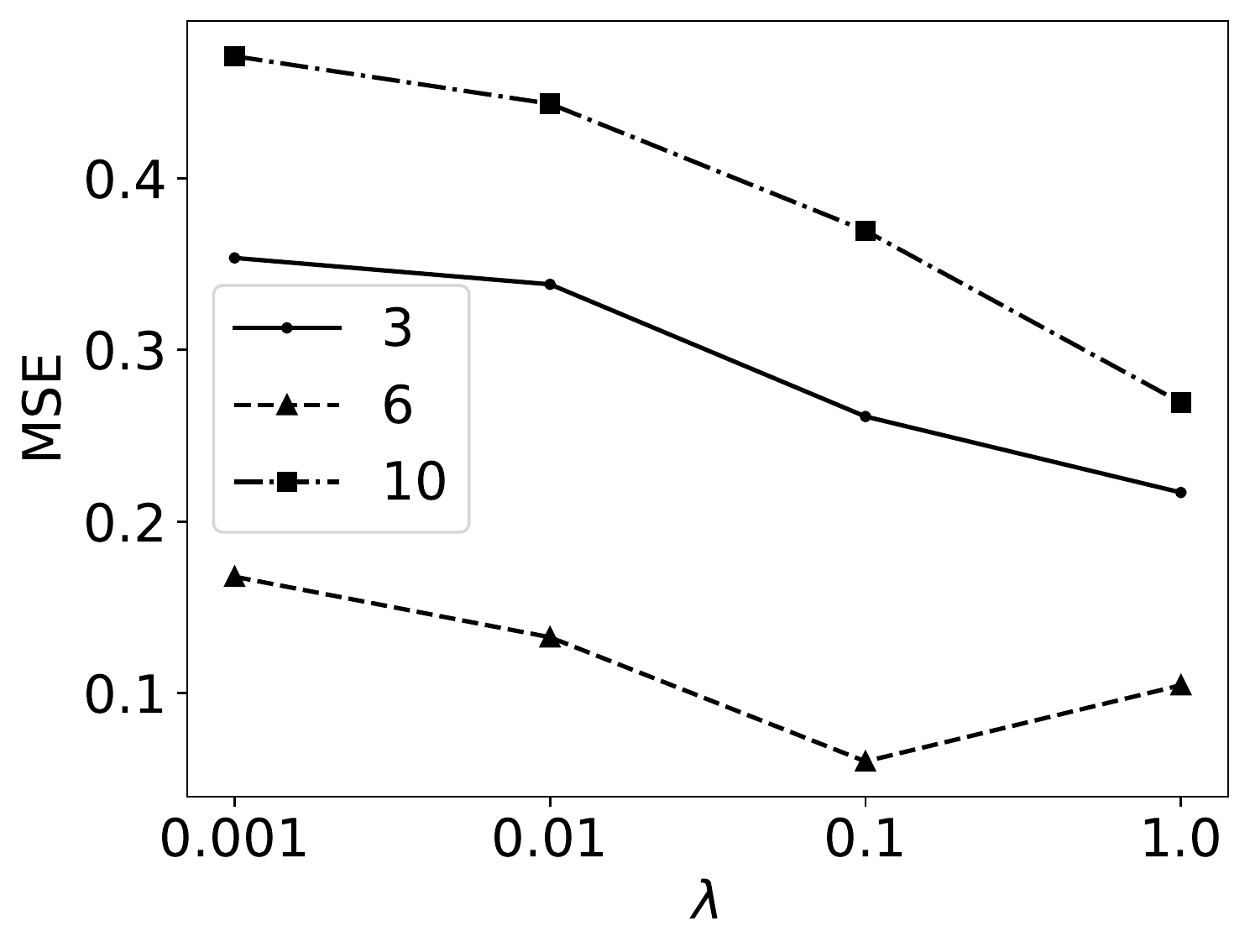}
    \caption*{(b) $N=1000$.}
  \end{minipage}
  \begin{minipage}[b]{0.3\linewidth}
    \centering
    \includegraphics[width=0.9\columnwidth]{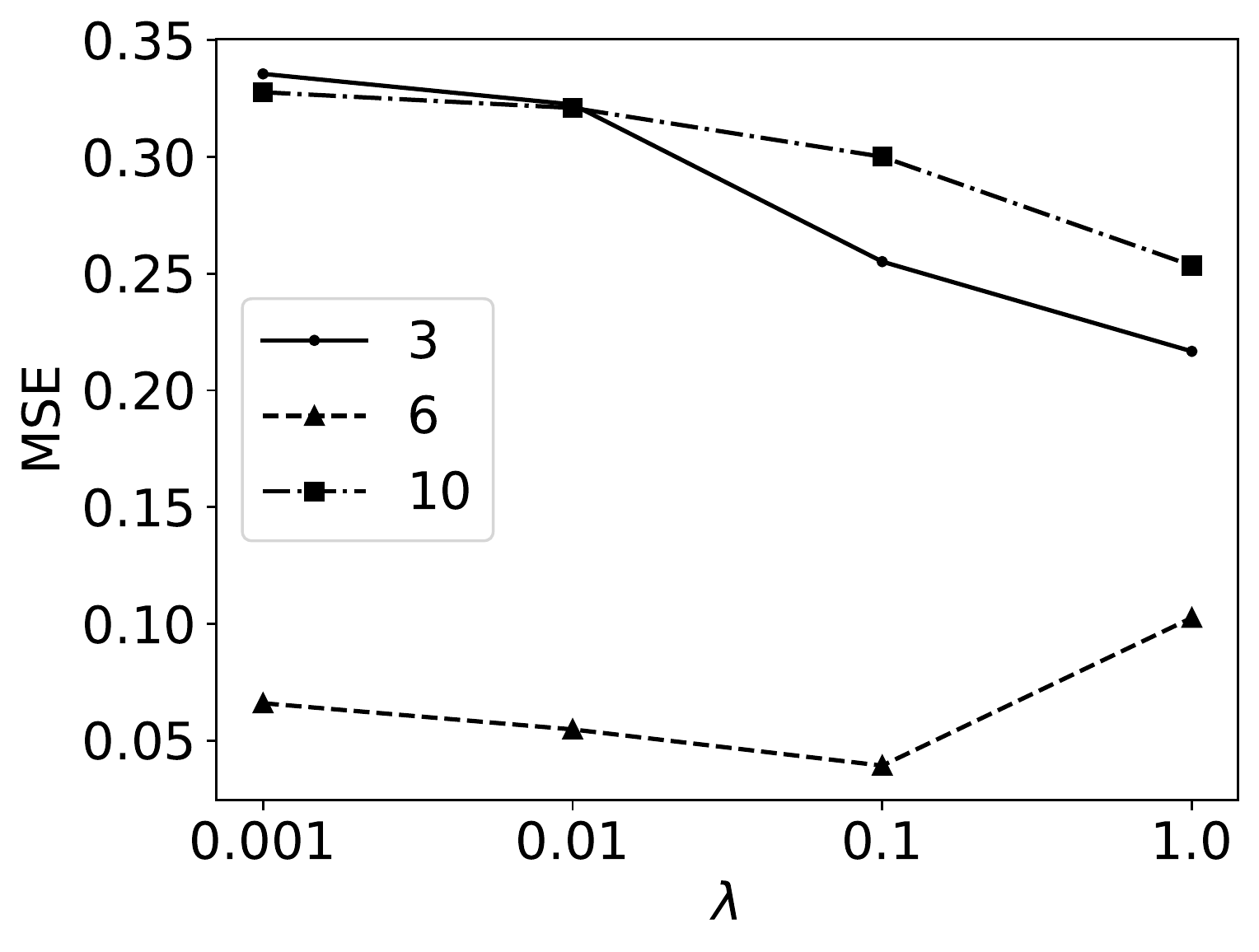}
    \caption*{(c) $N=2000$.}
  \end{minipage}
  \caption{Sensitivity Analysis for DGP-Q.}
  \label{fg:sens_dgpq}
\end{figure*}

\subsection{Orthogonal Series Ratio Estimator}
\textbf{Setting.}
As in the previous simulation, we consider CATE estimation by data combination, but now $H$ and $W$ depend on covariates.
Therefore, simply applying DSR must result in biased estimation.
We generate random samples of $O=(HY+(1-H)D,W,H,X)$ with a sample size $N=1000,2000,3000$ from the following DGP:
\begin{gather*}
    Y_D=\varsigma(\gamma'X)+0.4D\gamma'X+0.2\epsilon,\\
    P(D_1=1|X)=\varsigma(0.2\gamma'X+1),\hspace{5mm}P(D_0=1|X)=0,\\
    P(W=1|X)=\varsigma(0.1\gamma'X),\hspace{5mm}P(H=1|X)=\varsigma(-0.1\gamma'X).
\end{gather*}
where $X=(X_1,\ldots,X_5)'$ is a five-dimensional covariate vector with all elements having zero mean, and $\gamma$ is a five-dimensional vector with one in all elements.
The covariance matrix of $X$ takes one for diagonal elements, and the absolute value of non-diagonal elements is randomly chosen from $[0.1,0.3]$ except that $X_1$ is independent of all the other covariates.
The number of replication is set $1000$.

The estimand is CATE as a function of only the first covariate: $\theta_0(v)=E[Y_1-Y_0|X_1=v]=0.4v$.
We use the polynomial basis up to the third order to construct OSR, and the order is selected based on 5-fold CV in each replication.
ML estimators used for the nuisance estimation are random forest (RF), gradient boosting trees (GBT) and multi-layer perceptron (MLP).
ML estimators are implemented using scikit-learn 1.0.2, and we use the default value for all options and hyperparameters.
We use 5-fold cross-fitting to construct the orthogonal signals $\hat{U}$ and $\hat{T}$.
Unlike the previous one, regularization is not employed in this simulation.

\begin{table}[t]
\centering
\caption{Simulation Results for Orthogonal Series Estimator with Random Forest.}
\begin{tabular}{ccccccccccccc}
\hline
     &  & \multicolumn{3}{c}{$N=1000$} &  & \multicolumn{3}{c}{$N=2000$}  &  & \multicolumn{3}{c}{$N=3000$}  \\ \cline{3-5} \cline{7-9} \cline{11-13} 
$q$  &  & MSE   & Bias   & Width   &  & MSE    & Bias   & Width   &  & MSE     & Bias   & Width  \\ \hline
0.2  &  & 0.005 & -0.083 & 2.089   &  & 0.003  & -0.047 & 1.449   &  & 0.002   & -0.029 & 1.145  \\
0.4  &  & 0.011 & -0.011 & 2.667   &  & 0.005  & 0.000  & 1.725   &  & 0.004   & 0.008  & 1.317  \\
0.6  &  & 0.020 & 0.051  & 3.637   &  & 0.009  & 0.042  & 2.308   &  & 0.006   & 0.041  & 1.590  \\
0.8  &  & 0.038 & 0.130  & 7.665   &  & 0.016  & 0.092  & 4.497   &  & 0.010   & 0.081  & 3.327  \\
Mean &  & 0.096 & 0.023  & 48.238  &  & 0.027  & 0.024  & 9.691   &  & 0.012   & 0.028  & 5.250  \\
STD  &  & 1.153 & 0.140  & 936.180 &  & 0.325  & 0.090  & 110.607 &  & 0.073   & 0.067  & 38.138 \\ \cline{3-5} \cline{7-9} \cline{11-13} 
CVR  &  & \multicolumn{3}{c}{0.970} &  & \multicolumn{3}{c}{0.962} &  & \multicolumn{3}{c}{0.953} \\ \hline
\end{tabular}
\label{tb:osrrf}
\end{table}

\begin{table}[t]
\centering
\caption{Simulation Results for Orthogonal Estimator with Gradient Boosting Trees.}
\begin{tabular}{ccccccccccccc}
\hline
     &  & \multicolumn{3}{c}{$N=1000$}  &  & \multicolumn{3}{c}{$N=2000$}  &  & \multicolumn{3}{c}{$N=3000$}  \\ \cline{3-5} \cline{7-9} \cline{11-13} 
$q$  &  & MSE   & Bias   & Width    &  & MSE    & Bias   & Width   &  & MSE     & Bias   & Width  \\ \hline
0.2  &  & 0.009 & -0.097 & 2.790    &  & 0.003  & -0.026 & 1.476   &  & 0.002   & -0.003 & 1.066  \\
0.4  &  & 0.022 & 0.001  & 3.922    &  & 0.006  & 0.023  & 1.755   &  & 0.005   & 0.028  & 1.240  \\
0.6  &  & 0.040 & 0.077  & 5.911    &  & 0.011  & 0.060  & 2.189   &  & 0.007   & 0.060  & 1.511  \\
0.8  &  & 0.092 & 0.189  & 12.263   &  & 0.020  & 0.117  & 4.393   &  & 0.012   & 0.103  & 3.291  \\
Mean &  & 0.484 & 0.052  & 134.424  &  & 0.042  & 0.048  & 22.211  &  & 0.014   & 0.050  & 6.652  \\
STD  &  & 7.181 & 0.286  & 2497.353 &  & 0.645  & 0.094  & 428.023 &  & 0.071   & 0.066  & 66.374 \\ \cline{3-5} \cline{7-9} \cline{11-13} 
CVR  &  & \multicolumn{3}{c}{0.964} &  & \multicolumn{3}{c}{0.956} &  & \multicolumn{3}{c}{0.918} \\ \hline
\end{tabular}
\label{tb:osrgbt}
\end{table}

\begin{table}[t]
\centering
\caption{Simulation Results for Orthogonal Estimator with Multi-Layer Perceptron.}
\begin{tabular}{ccccccccccccc}
\hline
     &  & \multicolumn{3}{c}{$N=1000$}    &  & \multicolumn{3}{c}{$N=2000$}  &  & \multicolumn{3}{c}{$N=3000$}  \\ \cline{3-5} \cline{7-9} \cline{11-13} 
$q$  &  & MSE    & Bias   & Width     &  & MSE     & Bias  & Width   &  & MSE     & Bias   & Width  \\ \hline
0.2  &  & 0.004  & -0.020 & 1.485     &  & 0.002   & 0.002 & 1.026   &  & 0.002   & 0.010  & 0.819  \\
0.4  &  & 0.008  & 0.037  & 1.821     &  & 0.005   & 0.036 & 1.217   &  & 0.004   & 0.038  & 0.931  \\
0.6  &  & 0.014  & 0.080  & 2.402     &  & 0.007   & 0.070 & 1.481   &  & 0.006   & 0.060  & 1.116  \\
0.8  &  & 0.025  & 0.134  & 4.661     &  & 0.013   & 0.109 & 3.588   &  & 0.009   & 0.095  & 2.466  \\
Mean &  & 0.712  & 0.061  & 843.500   &  & 0.015   & 0.058 & 15.306  &  & 0.010   & 0.054  & 6.776  \\
STD  &  & 21.803 & 0.136  & 26423.507 &  & 0.079   & 0.070 & 209.077 &  & 0.047   & 0.056  & 58.231 \\ \cline{3-5} \cline{7-9} \cline{11-13} 
CVR  &  & \multicolumn{3}{c}{0.922}   &  & \multicolumn{3}{c}{0.916} &  & \multicolumn{3}{c}{0.905} \\ \hline
\end{tabular}
\label{tb:osrmlp}
\end{table}

\noindent\textbf{Results.}
Table \ref{tb:osrrf}, \ref{tb:osrgbt} and \ref{tb:osrmlp} present the simulation results when using RF, GBT and MLP for the nuisance estimation, respectively.
Bias in the tables denotes the estimation bias evaluated at $v=1$, Width is the largest width of the 95\% uniform confidence band, and CVR is the empirical coverage of the 95\% uniform confidence band, namely, the proportion of the times out of 1000 replications that the true target function $\theta_0(v)=0.4v$ is included in the estimated confidence band.
In addition to mean and standard deviation, we report the $q$-quantiles of MSE, Bias and Width for $q=0.2,0.4,0.6,0.8$ because OSR produces rare but extremely inaccurate estimates.
This instability may be due to the structure of the orthogonal signals, where the inverse of the estimated probability is included, and it has been pointed out in the literature \citep{singh2019automatic,chernozhukov2022automatic,chernozhukov2022riesz}.
Further discussion on this point can be found in Section \ref{sec:discussiondml}.

The results diverge for the different ML methods when the sample size is small, but it seems that the methods for nuisance estimates have less impact on the performance of OSR as the sample size grows.
The mean of MSE and the width of the confidence band get smaller as $N$ gets larger for all ML methods.
Although the mean of bias remains almost the same across sample sizes, the standard deviation decreases, which also indicates that an increase in sample size stabilises the estimation.
The empirical coverage is reasonably close to the nominal rate of 95\% for all the cases with RF and when GBT is used with $N=1000,2000$.
However, when using GBT with $N=3000$ and MLP with all the sample sizes, the empirical coverage deviates from the nominal coverage, possibly reflecting the slightly large bias and small width in these cases.
The large bias in GBT and MLP may be due to the fixed hyperparameters.
Although we used fixed hyperparameters to reduce the computational burden in this simulation, we could select hyperparameters based on CV in practice to lower the bias and obtain the confidence band with the correct coverage.

\section{Empirical Example}\label{sec:empirical}
In this section, we apply OSR to estimate a causal effect of participation in 401(k) on household assets.

\noindent\textbf{Setting.}
In the US, 401(k) is an employer-sponsored personal pension program first implemented in 1978.
It is tax-deferred to encourage household savings for their retirement.
Data from the US Census Bureau's 1991 Survey of Income and Program Participation (SIPP) has been used in \cite{poterba1994401,poterba1995401} and many subsequent studies to examine the effect of 401(k) participation on household savings.
The key technical challenge in estimating the causal effect of 401(k) participation is there are not enough covariates available in the SIPP data to explain the self-selected participation among those who are eligible to participate in 401(k).
\cite{poterba1994401,poterba1995401} argue that eligibility for the 401(k) program can be regarded as exogenous after conditioning on some important variables because whether an employer offered 401(k) would not affect people's job selection at least at the time the program just started, but they would instead make a decision based on other aspects of the job such as salary.
Adopting this argument, we can estimate LATE of 401(k) participation in household savings with 401(k) eligibility as IV and other variables related to job choice to control selection bias.

In this empirical example, we use the data analyzed in \cite{chernozhukov2004effects} and \cite{chernozhukov2018}, consisting of samples of 9915 households with the reference person of 25-64 years old, and at least one member is employed but no one is self-employed.
We use net financial assets ---defined as the sum of IRA (Individual Retirement Account) balances, 401(k) balances, checking accounts, US saving bonds, other interest-earning accounts in financial institutions, other interest-earning assets, stocks, and mutual funds minus non-mortgage debt--- as the outcome $Y$, a binary indicator for 401(k) eligibility as IV $Z$, and a binary indicator for 401(k) participation as the treatment $D$.
The covariate vector $X$ used in this analysis includes age, income, family size, years of education, marital status, two-earner status, defined benefit pension status, IRA participation status, and home ownership status.

We conduct analyses based on the usual one-sample LATE estimation as in Example \ref{ex:late} and two-sample LATE estimation explained in Section \ref{sec:2slate}.
The parameter of interest is LATE as a function of income.
For two-sample LATE estimation, we generate a dataset indicator $H$ such that $P(H=1|X=x)=\varsigma(0.1\gamma'\tilde{x})$, where $\tilde{x}$ is a vector of the covariates scaled so that the values fall within $[0,1]$.
We use the polynomial basis and the order is selected based on CV from $k=1,2,3$.
We use a relatively large number of partitions $G=20$ to mitigate the impact of random sample-splitting on the performance of OSR.
GBT is used for nuisance estimation.

\begin{figure*}[!t]
\centering
  \begin{minipage}[b]{0.49\linewidth}
    \centering
    \includegraphics[width=0.9\linewidth]{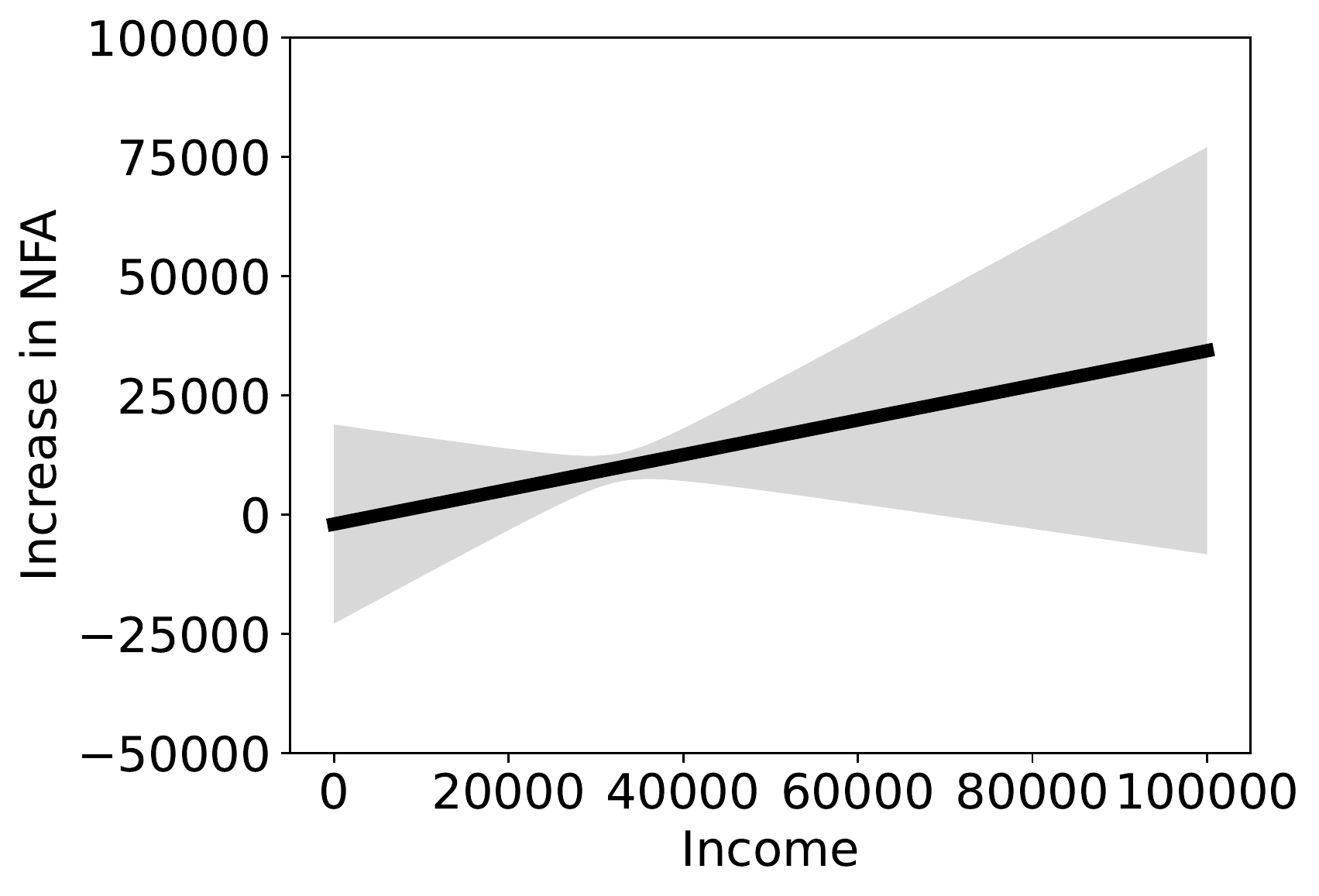}
    \caption*{(a) One-Sample Estimation.}
  \end{minipage}
  \begin{minipage}[b]{0.49\linewidth}
    \centering
    \includegraphics[width=0.9\linewidth]{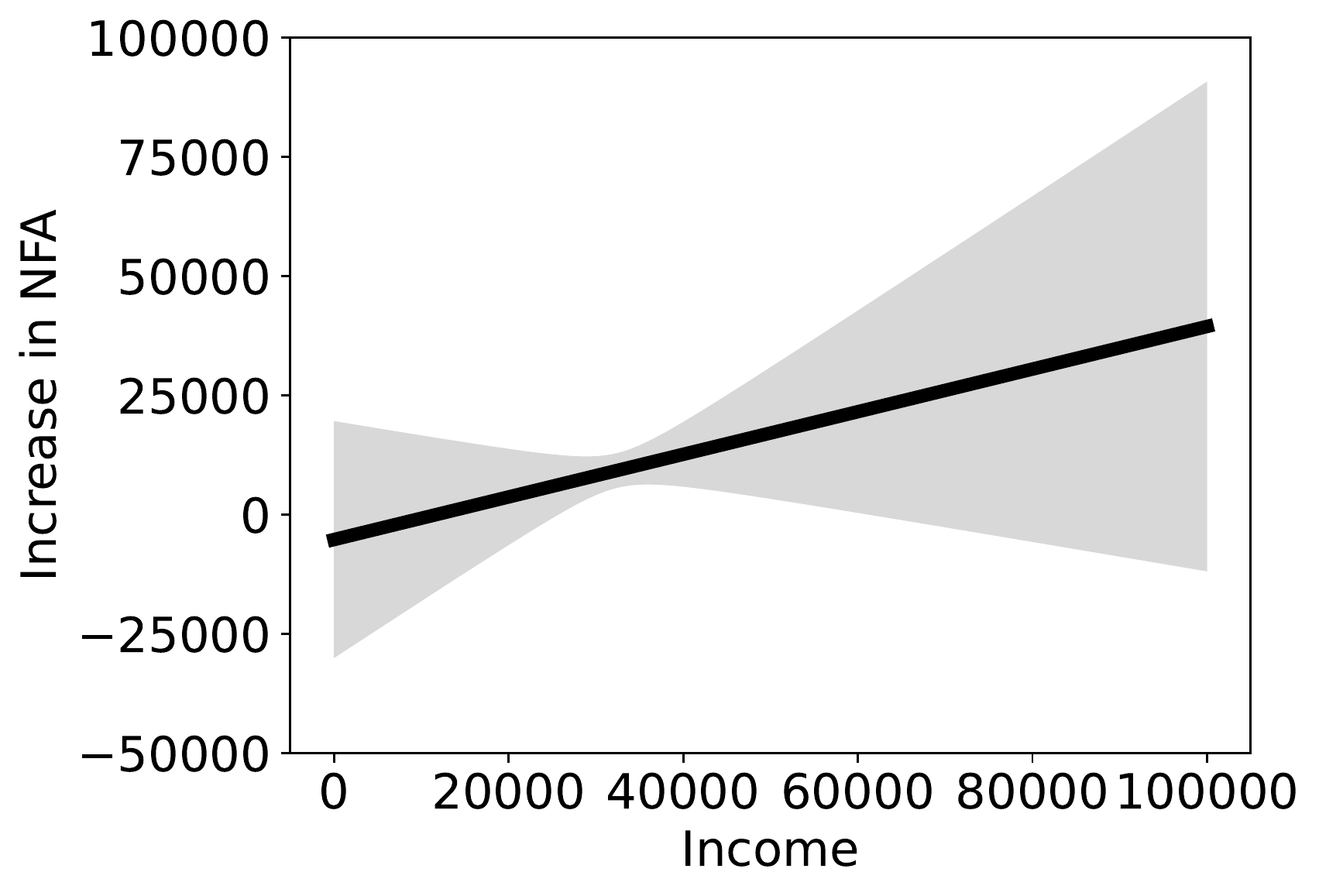}
    \caption*{(b) Two-Sample Estimation.}
  \end{minipage}
  \caption{Estimated LATE of 401(k) on Net Financial Assets Conditioned on Income Level and 95\% Uniform Confidence Band.}
  \label{fg:401k}
\end{figure*}

\noindent\textbf{Results.}
Figure \ref{fg:401k}(a) and (b) illustrate the estimated LATE function of income and its 95\% uniform confidence band constructed by one-sample estimation and two-sample estimation, respectively.
As can be seen, $k=1$ is selected in both one-sample and two-sample estimations.
The estimated LATE is linearly increasing in household income and the slope is approximately 0.36 in one-sample estimation and 0.45 in two-sample estimation.
Furthermore, the point estimate of LATE for households with annual income \$50000 is \$16225 in one-sample estimation and \$17141 in two-sample estimation, which are consistent with the analysis in \cite{ogburn2015} whose estimate is \$14910.
Statistically significant positive LATE is indicated for households with annual income \$24000-\$68000 in one-sample estimation, while LATE is significantly positive for households with income \$26000-61000 in two-sample estimation, reflecting the wider confidence band in two-sample estimation.

\section{Discussions}\label{sec:discussiondml}
This section discusses the limitation and future direction of the present study.
The first topic is model selection in the proposed framework.
As explained in Remark \ref{rm:cv}, we can perform model selection based on CV using the criterion (\ref{eq:cvcriterion}) only when $\zeta_0$ is strictly positive.
We explained the several examples where $\zeta_0$ is necessarily positive, and it is shown in the simulation study in Section \ref{sec:dsrsim} that CV using the criterion (\ref{eq:cvcriterion}) works well.
However, there are also situations where $\zeta_0$ can take both positive and negative values, such as LATE estimation and IDID.
Therefore, a model selection method for the general situation is key to increasing the practicality of the proposed framework.
Despite its practical importance, little attention has been paid to model selection in the treatment effects estimation \citep{schuler2018comparison,caron2020estimating}.
To the best of our knowledge, there exist only a few attempts in the literature to develop a flexible method for selecting the treatment effect model \citep{brookhart2006semiparametric,rolling2014model,saito2020}.
Although we may be able to extend the ideas of these previous studies for the CEFR problems, further research on model selection is essential to enhance the feasibility of causal inference methods.

One of the drawbacks of the proposed framework is the large variability found in the simulation of Section \ref{sec:dsrsim}.
It may be due to the structure of the orthogonal signals, in which the inverse of the estimated propensity score is used.
Inverse probability weighting (IPW) is known to suffer from unstable estimates especially when the propensity score is close to zero \citep{wooldridge2002,Wooldridge2007-vn,robins2007comment,Seaman2013-mw}.
This problem is especially acute in the data combination settings including two-sample estimation of LATE and IDID, where we have to perform four-class or eight-class classification to estimate propensity scores.
Although we can increase stability by trimming small probabilities, determining the optimal threshold is nontrivial \citep{lee2011}, and trimming can cause additional bias.
Recently developed \emph{automatic debiased machine learning (Auto-DML)} \citep{chernozhukov2022automatic} can be an effective solution to the problem because it avoids the estimation of propensity scores.
Auto-DML directly estimates the Riesz representer of the orthogonal signals rather than constructing it with the inverse of the estimated propensity score.
A much smaller variance of Auto-DML compared to the original DML has been empirically verified in numerical experiments \citep{singh2019automatic,chernozhukov2022riesz}.
Thus, extending the procedures and theory of the proposed framework to accommodate signals obtained by Auto-DML is a promising future direction.

The present study proposed a general and flexible framework for the CEFR problems, but more efficient estimation and inference may be possible in some specific settings.
For example, the orthogonal moment condition for LATE using the interaction term of $Y$ and $D$ has been proposed in \cite{singh2019automatic}, while OSR uses $Y$ and $D$ only separately.
Comparison of the efficiency of the method in \cite{singh2019automatic} and OSR is beyond the scope of this study, but intuitively, leveraging information expressed in the form of interaction of $Y$ and $D$ can improve efficiency.
However, the contribution of this study for offering the flexible inference framework in the data combination settings is significant, as the method in \cite{singh2019automatic} is not applicable to situations where $Y$ and $D$ are separately observed.

\bibliography{orthratio}
\newpage
\appendix
\section{Additional Examples}\label{sec:dmlexample2}
\subsection{Other Form of Ratio-Based Treatment Effects}
\cite{liang2020relative} has proposed the other form of ratio-based treatment effects such as:
\begin{align*}
    \theta_0(v)=\frac{E[Y_1-Y_0|V=v]}{E[Y_1+Y_0|V=v]}.
\end{align*}
In this case, we need Assumption \ref{as:rbcate} but with $E[Y_1+Y_0|V=v]\neq0$ for all $v\in\mathcal{V}$ instead of $E[Y_0|V=v]\neq0$ for all $v\in\mathcal{V}$ for identification.
We have that $\nu_0(v)=E[\mu_0(1,X)-\mu_0(0,X)|V=v]$ and $\zeta_0(v)=E[\mu_0(1,X)+\mu_0(0,X)|V=v]$, where $\mu_0(d,x)=E[Y|D=d,X=x]$.
The following signals satisfy the Neyman orthogonality:
\begin{align*}
    U(O,\eta_0)&=\mu_0(1,X)-\mu_0(0,X)+\frac{D(Y-\mu_0(1,X))}{\pi_0(X)}-\frac{(1-D)(Y-\mu_0(0,X))}{1-\pi_0(X)},\\
    T(O,\eta_0)&=\mu_0(1,X)+\mu_0(0,X)+\frac{D(Y-\mu_0(1,X))}{\pi_0(X)}+\frac{(1-D)(Y-\mu_0(0,X))}{1-\pi_0(X)},
\end{align*}
where $\pi_0(x)=E[D|X=x]$.

\subsection{Ratio-Based LATE}
We can also consider the ratio-based LATE, which is defined as:
\begin{align*}
    \theta_0(v)=\frac{E[Y_1|D_1>D_0,V=v]}{E[Y_0|D_1>D_0,V=v]}.
\end{align*}
If Assumption \ref{as:lateid} in addition to $E[Y_0|D_1>D_0,V=v]$ for $v\in\mathcal{V}$ hold, then
\begin{align*}
    \theta_0(v)=\frac{E[E[DY|Z=1,X]-E[DY|Z=0,X]|V=v]}{E[E[(1-D)Y|Z=0,X]-E[(1-D)Y|Z=1,X]|V=v]}.
\end{align*}
In this case, $\nu_0(v)=E[\mu_0(1,X)-\mu_0(0,X)|V=v]$ and $\zeta_0(v)=E[\pi_0(0,X)-\pi_0(1,X)|V=v]$, where $\mu_0(z,x)=E[DY|Z=z,X=x]$ and $\pi_0(z,x)=E[(1-D)Y|Z=z,X=x]$.
The following signals satisfy the Neyman orthogonality:
\begin{align*}
    U(O,\eta_0)&=\mu_0(1,X)-\mu_0(0,X)+\frac{Z(DY-\mu_0(1,X))}{\rho_0(X)}-\frac{(1-Z)(DY-\mu_0(0,X))}{1-\rho_0(X)},\\
    T(O,\eta_0)&=\pi_0(0,X)-\pi_0(1,X)+\frac{(1-Z)((1-D)Y-\pi_0(0,X))}{1-\rho_0(X)}-\frac{Z((1-D)Y-\pi_0(1,X))}{\rho_0(X)},
\end{align*}
where $\rho_0(x)=E[Z|X=x]$.

As in the previous subsection, we can consider the alternative form:
\begin{align*}
    \theta_0(v)=\frac{E[Y_1-Y_0|D_1>D_0,V=v]}{E[Y_1+Y_0|D_1>D_0,V=v]}.
\end{align*}
Under Assumption \ref{as:lateid} and $E[Y_1+Y_0|D_1>D_0,V=v]\neq0$ for $v\in\mathcal{V}$, then
\begin{align*}
    \theta_0(v)=\frac{E[E[Y|Z=1,X]-E[Y|Z=0,X]|V=v]}{E[E[(2D-1)Y|Z=1,X]-E[(2D-1)Y|Z=0,X]|V=v]}.
\end{align*}
In this case, $\nu_0(v)=E[\mu_0(1,X)-\mu_0(0,X)|V=v]$ and $\zeta_0(v)=E[\pi_0(1,X)-\pi_0(0,X)|V=v]$, where $\mu_0(z,x)=E[Y|Z=z,X=x]$ and $\pi_0(z,x)=E[(2D-1)Y|Z=z,X=x]$.
The following signals satisfy the Neyman orthogonality:
\begin{align*}
    U(O,\eta_0)&=\mu_0(1,X)-\mu_0(0,X)+\frac{Z(Y-\mu_0(1,X))}{\rho_0(X)}-\frac{(1-Z)(Y-\mu_0(0,X))}{1-\rho_0(X)},\\
    T(O,\eta_0)&=\pi_0(1,X)-\pi_0(0,X)+\frac{Z((2D-1)Y-\pi_0(1,X))}{\rho_0(X)}-\frac{(1-Z)((2D-1)Y-\pi_0(0,X))}{1-\rho_0(X)},
\end{align*}
where $\rho_0(x)=E[Z|X=x]$.

\subsection{Two-Sample LATE}\label{sec:2slate}
It is obvious from the identification result of difference-based LATE that the two-sample estimation is possible.
Let the observed vector be $O=(HY+(1-H)D,Z,H,X)$, where $H$ is a binary dataset indicator, and suppose that $Y_1,Y_0,D_1,D_0\independent H|X$ in addition to Assumption \ref{as:lateid}.
Then, LATE is identified as:
\begin{align*}
    \theta_0(v)=\frac{E[E[Y|H=1,Z=1,X]-E[Y|H=1,Z=0,X]|V=v]}{E[E[D|H=0,Z=1,X]-E[D|H=0,Z=0,X]|V=v]}.
\end{align*}
In this case, $\nu_0(v)=E[\mu_0(1,X)-\mu_0(0,X)|V=v]$ and $\zeta_0(v)=E[\pi_0(1,X)-\pi_0(0,X)|V=v]$, where $\mu_0(z,x)=E[Y|H=1,Z=z,X=x]$ and $\pi_0(z,x)=E[D|H=0,Z=z,X=x]$.
The following signals satisfy the Neyman orthogonality:
\begin{align*}
    U(O,\eta_0)&=\mu_0(1,X)-\mu_0(0,X)+\frac{HZ(Y-\mu_0(1,X))}{\rho_0(1,1,X)}-\frac{H(1-Z)(Y-\mu_0(0,X))}{\rho_0(1,0,X)},\\
    T(O,\eta_0)&=\pi_0(1,X)-\pi_0(0,X)+\frac{(1-H)Z(D-\pi_0(1,X))}{\rho_0(0,1,X)}-\frac{(1-H)(1-Z)(D-\pi_0(0,X))}{\rho_0(0,0,X)},
\end{align*}
where $\rho_0(h,z,x)=P(H=h,Z=z|X=x)$.

Although two-sample LATE estimation is similar to the data combination in Example \ref{ex:dcte}, they are different in the observed variables and the underlying assumptions.
The significant difference is that we must observe $Z$ in two-sample LATE estimation, whereas we only need to know that the treatment regimes $P(Z_w=1|X)$ are different according to $w$ in Example \ref{ex:dcte}.

\subsection{Two-Sample IDID}
As shown in \cite{Ye2020InstrumentedD}, two-sample estimation is also possible in the IDID setting.
The observable can be expressed as $O=(HY+(1-H)D,Z,W,H,X)$.
If $D_{zw},Y_{1w}-Y_{0w},Y_{01}-Y_{00}\independent H|X$ for $w,z=0,1$, then CATE is identified as:
\begin{align*}
    \theta_0(v)=\frac{E[\mu_0(1,1,X)-\mu_0(0,1,X)-\mu_0(1,0,X)+\mu_0(0,0,X)|V=v]}{E[\pi_0(1,1,X)-\pi_0(0,1,X)-\pi_0(1,0,X)+\pi_0(0,0,X)|V=v]},
\end{align*}
where $\mu_0(z,x)=E[Y|H=1,W=w,Z=z,X=x]$ and $\pi_0(z,x)=E[D|H=0,W=w,Z=z,X=x]$.
In this case,
\begin{align*}
    \nu_0(v)&=\sum_{w,z\in\{0,1\}^2}(-1)^{w+z}E[\mu_0(w,z,X)|V=v],\\
    \nu_0(v)&=\sum_{w,z\in\{0,1\}^2}(-1)^{w+z}E[\pi_0(w,z,X)|V=v].
\end{align*}
The following signals satisfy the Neyman orthogonality:
\begin{align*}
    U(O,\eta_0)&=\sum_{w,z\in\{0,1\}^2}(-1)^{w+z}\left(\mu_0(w,z,X)+\frac{1_{H=1}1_{W=w}1_{Z=z}(Y-\mu_0(w,z,X))}{\rho_0(1,w,z,X)}\right),\\
    T(O,\eta_0)&=\sum_{w,z\in\{0,1\}^2}(-1)^{w+z}\left(\pi_0(w,z,X)+\frac{1_{H=0}1_{W=w}1_{Z=z}(D-\pi_0(w,z,X))}{\rho_0(0,w,z,X)}\right),
\end{align*}
where $\rho_0(h,w,z,x)=P(H=h,W=w,Z=z|X=x)$.
We can obtain the estimate $\hat{\rho}(h,w,z,x)$ by implementing an eight-class classification.
However, it becomes easier for $\hat{\rho}(h,w,z,x)$ to take small values as the number of classes increases, leading to the unstable estimation in practice.

\section{Proofs}
\subsection{Proofs of Results in Section \ref{sec:asymtheory}}
\begin{lemma}[LLN for Matrices]\label{le:matlln}
Let $Q_i, i\in[N]$ be independent not necessarily identically distributed and symmetric non-negative $K\times K$ matrices such that $K\geq2$ and $\|Q_i\|\leq M$ $a.s.$
Let $Q=E_N[E[Q_i]]$ and $\hat{Q}=E_N[Q_i]$.
Then,
\begin{align*}
    E[\|\hat{Q}-Q\|]\leq\sqrt{\frac{M(1+\|Q\|)\log N}{N}}.
\end{align*}
In particular, if $Q_i=p_ip_i't_i$ with $\|p_i\|\leq\xi_k$ and $|t_i|\leq1$, then
\begin{align*}
    E[\|\hat{Q}-Q\|]\leq\sqrt{\frac{\xi_k^2(1+\|Q\|)\log N}{N}}.
\end{align*}
\end{lemma}
See \cite{rudelson1999random,tropp2015introduction,belloni2015some} for the proof.

\begin{lemma}\label{le:regularQ}
Let Assumption \ref{as:basis_eigen} and \ref{as:boundedfunc} hold.
Then, all eigenvalues of $Q=E[p(V)p(V)'\zeta_0(V)]$ are bounded above and away from zero uniformly over $k$.
\end{lemma}
\begin{proof}[Proof of Lemma \ref{le:regularQ}]
Define $\underline{\zeta}^2:=\inf_{v\in\mathcal{V}}\zeta_0(v)^2$.
For any $a\in\{a\in\mathbb{R}^k:\|a\|\neq0\|\}$,
\begin{align*}
    (a'E[p(V)p(V)'\zeta_0(V)]a)^2&=E[(a'p(V))^2\zeta(V)]^2\geq E[(a'p(V))^2]^2\underline{\zeta}^2>0.
\end{align*}
The last inequality follows from $E[(a'p(V))^2]>0$ Assumption \ref{as:basis_eigen} and $\underline{\zeta}^2>0$ by Assumption \ref{as:boundedfunc}.
\end{proof}

\begin{lemma}[No Effect of First-Stage Error]\label{le:analograte}
Under Assumption \ref{as:smallbias},
\begin{align*}
    \sqrt{N}\|E_N[p_i(U_i(\hat{\eta})-U_i(\eta_0))]\|&=O_P(B_N+\Lambda_N)=o(1),\\
    \sqrt{N}\|E_N[p_i(T_i(\hat{\eta})-T_i(\eta_0))]\|&=O_P(B_N+\Lambda_N)=o(1).
\end{align*}
\end{lemma}
\begin{proof}[Proof of Lemma \ref{le:analograte}]
Lemma \ref{le:analograte} follows from Assumption \ref{as:smallbias} and the proof of Lemma A.3 in \cite{semenova2021debiased}.
\end{proof}

\begin{proof}[Proof of Lemma \ref{le:prate}]
\textbf{(a)} By Assumption \ref{as:basis_rate} and Lemma \ref{le:matlln}, $\|\hat{Q}-Q\|\rightarrow_P0$ as $N\rightarrow\infty$.
Therefore, all eigenvalues of $\hat{Q}$ are also bounded away from zero with probability approaching one by Lemma $\ref{le:regularQ}$.
Thus,
\begin{align*}
    \|\hat{\beta}-\beta_k\|&=\|\hat{Q}^{-1}E_N[p_i\hat{U}_i]-\beta_0\|=\|\hat{Q}^{-1}E_N[p_i(\hat{U}_i-\hat{T}_i\theta_{ki})]\|\lesssim\|E_N[p_i(\hat{U}_i-\hat{T}_i\theta_{ki})]\|.
\end{align*}
We can decompose the estimated orthogonal scores as:
\begin{align*}
    \hat{U}=(\hat{U}-U_0)+\nu_0+\varepsilon_U,\hspace{5mm}\hat{T}=(\hat{T}-T_0)+\zeta_0+\varepsilon_T.
\end{align*}
Using this decomposition and noting that $\theta_k=\theta_0-r$, we have
\begin{align*}
    E_N[p_i(\hat{U}_i-\hat{T}_i\theta_{ki})]&=E_N[p_i(\hat{U}_i-U_{0i})]+E_N[p_i\theta_{ki}(\hat{T}_i-T_{0i})]\\
    &\qquad+E_N[p_i\varepsilon_{Ui}]-E_N[p_i\theta_{0i}\varepsilon_{Ti}]+E_N[p_i\zeta_{0i}r_i]+E_N[p_i\varepsilon_{Ti}r_i]\\
    &=:J_1+J_2+J_3-J_4+J_5+J_6.
\end{align*}
Therefore, it suffices to derive the bound on the norm of these six terms by the triangle inequality.

The first-stage errors $J_1$ and $J_2$ are bounded as
\begin{align*}
    \|J_1\|+\|J_2\|&\leq\|E_N[p_i(\hat{U}_i-U_{0i})]\|+\|E_N[p_i\theta_{0i}(\hat{T}_i-T_{0i})]\|+\|E_N[p_ir_i(\hat{T}_i-T_{0i})]\|\\
    &=o_P(N^{-1/2}(1+l_kr_k)),
\end{align*}
by Lemma \ref{le:analograte}, Assumption \ref{as:approxerror} and Assumption \ref{as:boundedfunc}.
For $J_3$ and $J_4$, we have
\begin{align*}
    \|J_3\|+\|J_4\|\lesssim\|E_N[p_i\varepsilon_i]\|
\end{align*}
since $\theta_0=\nu_0/\zeta_0$ is uniformly bounded by Assumption \ref{as:boundedfunc}.
Similarly for $J_5$ and $J_6$, we have
\begin{align*}
    \|J_5\|+\|J_6\|\lesssim\|E_N[p_ir_i]\|
\end{align*}
by Assumption \ref{as:samplingerror} and \ref{as:boundedfunc}.
As shown in the proof of Theorem 4.1 in \cite{belloni2015some}, the bound is given as
\begin{align*}
    \|E_N[p_i\varepsilon_i]\|+\|E_N[p_ir_i]\|\lesssim_P\sqrt{\frac{k}{N}}+\left(\sqrt{\frac{k}{N}}l_kr_k\land\frac{\xi_kr_k}{\sqrt{N}}\right).
\end{align*}
This completes the proof of Lemma \ref{le:prate} (a).

\noindent\textbf{(b)}
Using the same decomposition as in part (a), we obtain
\begin{align*}
    \sqrt{N}\alpha'(\hat{\beta}-\beta_0)&=\sqrt{N}\alpha'\hat{Q}^{-1}(J_1+J_2+J_3-J_4+J_5+J_6)\\
    &=\alpha'Q^{-1}\mathbb{G}_N[p_i(\varepsilon_{Ui}-\theta_{0i}\varepsilon_{Ti})]+\alpha'(\hat{Q}^{-1}-Q^{-1})\mathbb{G}_N[p_i(\varepsilon_{Ui}-\theta_{0i}\varepsilon_{Ti})]\\
    &\qquad+\sqrt{N}\alpha'Q^{-1}(J_1+J_2)+\sqrt{N}\alpha'(\hat{Q}^{-1}-Q^{-1})(J_1+J_2)\\
    &\qquad+\sqrt{N}\alpha'Q^{-1}(J_5+J_6)+\sqrt{N}\alpha'(\hat{Q}^{-1}-Q^{-1})(J_5+J_6)\\
    &=:\alpha'Q^{-1}\mathbb{G}_N[p_i(\varepsilon_{Ui}-\theta_{0i}\varepsilon_{Ti})]+L_1+L_2+L_3+L_4+L_5\\
    &=:\alpha'Q^{-1}\mathbb{G}_N[p_i(\varepsilon_{Ui}-\theta_{0i}\varepsilon_{Ti})]+R_N(\alpha).
\end{align*}
By Assumption \ref{as:boundedfunc}, we have that
\begin{align*}
    |L_1|\lesssim|\alpha'(\hat{Q}^{-1}-Q^{-1})\mathbb{G}_N[p_i\varepsilon_i]|.
\end{align*}
and its bound can be obtained as
\begin{align*}
    |\alpha'(\hat{Q}^{-1}-Q^{-1})\mathbb{G}_N[p_i\varepsilon_i]|\lesssim_P\sqrt{\frac{\xi_k^2\log N}{N}},
\end{align*}
by the proof of Lemma 4.1 in \cite{belloni2015some}.
By Lemma \ref{le:matlln}, \ref{le:regularQ}, \ref{le:analograte} and Assumption \ref{as:basis_rate},
\begin{align*}
    |L_2|&\lesssim_P\|\alpha\|\|Q^{-1}\|\|\sqrt{N}(J_1+J_2)\|=o(1),\\
    |L_3|&\lesssim_P\|\alpha\|\|\hat{Q}^{-1}-Q^{-1}\|\|\sqrt{N}(J_1+J_2)\|\lesssim_P\sqrt{\frac{\xi_k^2\log N}{N}}o(1)=o(1).
\end{align*}
For $|L_4|$, we have that
\begin{align*}
    |L_4|&\lesssim\left|\sqrt{N}E_N[\alpha'p_ir_i]\right|\lesssim_P E[(\alpha'p_ir_i)^2]^{1/2}\lesssim l_kr_k,
\end{align*}
where the first wave inequality follows from Assumption \ref{as:samplingerror} and \ref{as:boundedfunc}, the second inequality follows from Chebyshev's inequality since $E[p_ir_i]=0$, and the last wave inequality follows from Assumption \ref{as:approxerror}.
Also, for $|L_5|$, we have that
\begin{align*}
    |L_5|&\lesssim\left\|\hat{Q}^{-1}-Q^{-1}\right\|\left\|\sqrt{N}E_N[p_ir_i]\right\|\lesssim_P\sqrt{\frac{\xi_k^2\log N}{N}}\left(l_kr_k\sqrt{k}\land\xi_kr_k\right),
\end{align*}
which is shown in the proof of Lemma 4.1 in \cite{belloni2015some}.
Combining the bounds on $L_j$ for $j\in[5]$ gives the linearization result.
\end{proof}

\begin{proof}[Proof of Theorem \ref{th:pnormal}]
Theorem \ref{th:pnormal} can be proved by using Lemma \ref{le:prate} and applying the proof of Theorem 4.2 in \cite{belloni2015some}.
\end{proof}

\begin{proof}[Proof of Lemma \ref{le:urate}]
\noindent\textbf{(a)} The similar decomposition as in the proof of Lemma \ref{le:prate} (b) gives
\begin{align*}
    R_N(\alpha(v))&=\alpha(v)'(\hat{Q}^{-1}-Q^{-1})\mathbb{G}_N[p_i(\varepsilon_{Ui}+\theta_{0i}\varepsilon_{Ti})]\\
    &\qquad+\sqrt{N}\alpha(v)'Q^{-1}(J_1+J_2)+\sqrt{N}\alpha(v)'(\hat{Q}^{-1}-Q^{-1})(J_1+J_2)\\
    &\qquad+\sqrt{N}\alpha(v)'Q^{-1}(J_5+J_6)+\sqrt{N}\alpha(v)'(\hat{Q}^{-1}-Q^{-1})(J_5+J_6)\\
    &=:L_1(v)+L_2(v)+L_3(v)+L_4(v)+L_5(v).
\end{align*}
By the similar argument as in the proof of Lemma \ref{le:prate} (b), we have that
\begin{align*}
    |L_1(v)|\lesssim\left|\alpha(v)'(\hat{Q}^{-1}-Q^{-1})\mathbb{G}_N[p_i\varepsilon_i]\right|.
\end{align*}
According to the proof of Lemma 4.2 in \cite{belloni2015some},
\begin{align*}
    \sup_{v\in\mathcal{V}}\left|\alpha(v)'(\hat{Q}^{-1}-Q^{-1})\mathbb{G}_N[p_i\varepsilon_i]\right|\lesssim_P N^{1/m}\sqrt{\frac{\xi_k^2\log^2N}{N}}.
\end{align*}
By Lemma \ref{le:matlln}, \ref{le:regularQ}, \ref{le:analograte} and Assumption \ref{as:basis_rate},
\begin{align*}
    \sup_{v\in\mathcal{V}}|L_2(v)|&\leq\sup_{v\in\mathcal{V}}\|\alpha(v)\|\left\|Q^{-1}\right\|\left\|\sqrt{N}(J_1+J_2)\right\|=o(1),\\
    \sup_{v\in\mathcal{V}}|L_3(v)|&\leq\sup_{v\in\mathcal{V}}\|\alpha(v)\|\left\|\hat{Q}^{-1}-Q^{-1}\right\|\left\|\sqrt{N}(J_1+J_2)\right\|\lesssim_P\sqrt{\frac{\xi_k^2\log N}{N}}o(1)=o(1).
\end{align*}
For $L_4(v)$ and $L_5(v)$, we have that
\begin{align*}
    \sup_{v\in\mathcal{V}}|L_4(v)|&\lesssim\sup_{v\in\mathcal{V}}\left|\sqrt{N}\alpha(v)'Q^{-1}E_N[p_ir_i]\right|\lesssim_P l_kr_k\sqrt{\log k},\\
    \sup_{v\in\mathcal{V}}|L_5(v)|&\lesssim\sup_{v\in\mathcal{V}}\left|\sqrt{N}\alpha(v)'(\hat{Q}^{-1}-Q^{-1})E_N[p_ir_i]\right|\lesssim_P\sqrt{\frac{\xi_k^2\log N}{N}}\left(l_kr_k\sqrt{k}\land\xi_kr_k\right),
\end{align*}
where the first inequality in both lines follows from Assumption \ref{as:samplingerror} and \ref{as:boundedfunc}, and the both bounds are given by Lemma 4.2 in \cite{belloni2015some}.
Combining the bounds on $L_j(v)$ for $j\in[5]$ gives the linearization result.

\noindent\textbf{(b)}
Lemma \ref{le:urate} (b) follows from Lemma \ref{le:urate} (a) and the proof of Theorem 4.3 in \cite{belloni2015some}.
\end{proof}

\begin{proof}[Proof of Theorem \ref{th:strongaprx}]
Theorem \ref{th:strongaprx} can be proved by using Lemma \ref{le:urate} and applying the proof of Theorem 4.4 in \cite{belloni2015some}.
\end{proof}

\begin{proof}[Proof of Theorem \ref{th:matest}]
Theorem \ref{th:matest} can be proved by applying the proof of Theorem 3.3 in \cite{semenova2021debiased}.
\end{proof}

\begin{proof}[Proof of Theorem \ref{th:bootstrap}]
Theorem \ref{th:bootstrap} follows from the proof of Theorem 3.4 in \cite{semenova2021debiased}.
See also \cite{chernozhukov2013intersection} for details.
\end{proof}

\begin{proof}[Proof of Theorem \ref{th:confband}]
Theorem \ref{th:confband} follows from the proof of Theorem 3.5 in \cite{semenova2021debiased}.
See also \cite{belloni2015some} for details.
\end{proof}

\subsection{Proofs of Results in Section \ref{sec:app_dml}}
\begin{lemma}[Neyman Orthogonality of the Doubly Robust Signal]\label{le:orthosignal}
Consider the following doubly robust type signal:
\begin{align*}
    U(O,\eta_0)&=\mu_0(1,X)+\frac{Z(Y-\mu_0(1,X))}{\rho_0(X)},
\end{align*}
where $Z$ is a binary variable, $\mu_0(z,x)=E[Y|Z=z,X=x]$ is an outcome regression function and $\rho_0(x)=E[Z|X=x]$ is a propensity score.
Also, define $\nu_0(v)=E[\mu_0(1,X)|V=v]$.
Then, the above signal satisfies the Neyman orthogonality, namely, for all $v\in\mathcal{V}$,
\begin{align*}
    \partial_s E[U(O,\eta_0+s(\eta-\eta_0))-\nu_0(V)|V=v]|_{s=0}=0.
\end{align*}
\end{lemma}
\begin{proof}[Proof of Lemma \ref{le:orthosignal}]
If we can interchange differentiation and integration, we have, for $\Psi(s):=\partial_s E[U(O,\eta_0+s(\eta-\eta_0))-\nu_0(V)|V=v]$, that
\begin{align*}
    \Psi(s)&=\partial_s E\left[\left.s(\mu-\mu_0)-\frac{s\rho_0(\mu-\mu_0)}{\rho_0+s(\rho-\rho_0)}\right|V\right]\\
    &=E\left[\left.\mu-\mu_0-\frac{\rho_0(\mu-\mu_0)(\rho_0+s(\rho-\rho_0))-s\rho_0(\mu-\mu_0)(\rho-\rho_0)}{(\rho_0+s(\rho-\rho_0))^2}\right|V\right]\\
    &=E\left[\left.\mu-\mu_0-\frac{\rho_0^2(\mu-\mu_0)}{(\rho_0+s(\rho-\rho_0))^2}\right|V\right].
\end{align*}
Thus, $\Psi(0)=0$, which concludes the proof.
\end{proof}

\begin{proof}[Proof of Theorem \ref{th:mclate}]
\noindent\textbf{(a)}
We have that
\begin{align*}
    \nu_0(v)&=E[D_1Y_1+(1-D_1)Y_0|V=v]-E[D_0Y_1+(1-D_0)Y_0|V=v]\\
    &=E[(D_1-D_0)(Y_1-Y_0)|V=v]\\
    &=E[D_1-D_0|V=v]E[Y_1-Y_0|D_1-D_0=1,V=v],
\end{align*}
where the first equality holds by Assumption \ref{as:lateid}.1, and the last inequality follows from that $D_1-D_0\in\{0,1\}$ by Assumption \ref{as:lateid}.4.
Likewise, we have that
\begin{align*}
    \zeta_0(v)&=E[D_1-D_0|V=v],
\end{align*}
by Assumption \ref{as:lateid}.1.
Lastly, because $E[D_1-D_0|V=v]\neq0$ for all $v\in\mathcal{V}$ by Assumption \ref{as:lateid}.3, $\nu_0(v)/\zeta_0(v)$ is well-defined and it equals to
\begin{align*}
    \frac{\nu_0(v)}{\zeta_0(v)}&=\frac{E[D_1-D_0|V=v]E[Y_1-Y_0|D_1-D_0=1,V=v]}{E[D_1-D_0|V=v]}\\
    &=E[Y_1-Y_0|D_1>D_0,V=v]\\
    &=\theta_0(v),
\end{align*}
where the second equality holds since $D_1-D_0=1$ implies that $D_1>D_0$.

\noindent\textbf{(b)}
The moment conditions in part (b) immediately follow from the definition of the regression functions $\mu_0,\pi_0$ and the propensity score $\rho_0$.

\noindent\textbf{(c)}
The Neyman orthogonality of the signals $U$ and $T$ follows from Lemma \ref{le:orthosignal}.
\end{proof}

\begin{proof}[Proof of Corollary \ref{cr:late}]
We can decompose $U(O,\eta)-U(O,\eta_0)$ as
\small\begin{align*}
    U(O,\eta)-U(O,\eta_0)&=(\mu(1,X)-\mu_0(1,X))\left(1-\frac{Z}{\rho_0(X)}\right)-(\mu(0,X)-\mu_0(0,X))\left(1-\frac{1-Z}{1-\rho_0(X)}\right)\\
    &\quad+(\rho_0(X)-\rho(X))(Y-\mu_0(1,X))\frac{Z}{\rho(X)\rho_0(X)}\\
    &\qquad-(\rho_0(X)-\rho(X))(Y-\mu_0(0,X))\frac{1-Z}{(1-\rho(X))(1-\rho_0(X))}\\
    &\quad+(\rho_0(X)-\rho(X))(\mu_0(1,X)-\mu(1,X))\frac{D}{\rho(X)\rho_0(X)}\\
    &\qquad-(\rho_0(X)-\rho(X))(\mu_0(0,X)-\mu(0,X))\frac{1-D}{(1-\rho(X))(1-\rho_0(X))}\\
    &=:S_1+S_1'+S_2+S_2'+S_3+S_3'.
\end{align*}\normalsize
In what follows, we will bound $S_1$, $S_2$ and $S_3$ as the bounds on $S_1'$, $S_2'$ and $S_3'$ follow from a similar argument.

\noindent\textbf{Bound on $B_N$.}
By the law of iterated expectation, we have
\begin{align*}
    E[p(V)S_1]=E\left[p(V)(\mu(1,X)-\mu_0(1,X))E\left[\left.1-\frac{Z}{\rho_0(X)}\right|X\right]\right]=0,
\end{align*}
and likewise,
\begin{align*}
    E[p(V)S_2]=E\left[p(V)\frac{\rho_0(X)-\rho(X)}{\rho(X)\rho_0(X)}E[Z(Y-\mu_0(1,X))|X]\right]=0.
\end{align*}
For $S_3$, we have that
\begin{align*}
    \|E[p(V)S_3]\|^2\leq\overline{C}^3\|E[p(V)(\rho_0(X)-\rho(X))(\mu_0(1,X)-\mu(1,X))]\|^2,
\end{align*}
where we use the fact that
\begin{align}
    E[Z^2\rho(X)^{-2}\rho_0(X)^{-2}|X]=E[Z|X]\rho(X)^{-2}\rho_0(X)^{-2}\leq\overline{C}^3.\label{eq:sqbound}
\end{align}
Therefore, for the case of bounded basis, we can bound $\|E[p(V)S_3]\|^2$ as
\begin{align*}
    \|E[p(V)S_3]\|^2&\lesssim\|E[p(V)(\rho_0(X)-\rho(X))(\mu_0(1,X)-\mu(1,X))]\|^2\\
    &\lesssim\sum_{j=1}^k E[p_j(V)(\rho_0(X)-\rho(X))(\mu_0(1,X)-\mu(1,X))]^2\\
    &\leq k\mathbf{m}_{N,2}^2\mathbf{r}_{N,2}^2,
\end{align*}
where the last inequality follows by Cauchy–Schwarz inequality.
Since the same bound can be derived for $S_1'$, $S_2'$ and $S_3'$, it follows that
\begin{align*}
    \|E[p(V)(U-U_0)]\|\lesssim\sqrt{k}\mathbf{m}_{N,2}\mathbf{r}_{N,2}.
\end{align*}
The almost same argument derives the bound
\begin{align*}
    \|E[p(V)(T-T_0)]\|\lesssim\sqrt{k}\mathbf{p}_{N,2}\mathbf{r}_{N,2}.
\end{align*}
since the signals $U$ and $T$ have the same structure.
Thus, we have that
\begin{align*}
    B_N\lesssim\sqrt{kN}\mathbf{r}_{N,2}(\mathbf{m}_{N,2}+\mathbf{p}_{N,2})=o(1).
\end{align*}
When the basis is not bounded, we can alternatively bound $\|E[p(V)S_3]\|^2$ as
\begin{align*}
    \|E[p(V)S_3]\|^2&\lesssim\|E[p(V)(\rho_0(X)-\rho(X))(\mu_0(1,X)-\mu(1,X))]\|^2\\
    &\lesssim\sum_{j=1}^k E[p_j(V)(\rho_0(X)-\rho(X))^2(\mu_0(1,X)-\mu(1,X))^2]^2\\
    &\leq k\mathbf{m}_{N,2\omega}^2\mathbf{r}_{N,2\psi}^2,
\end{align*}
where the last inequality follows by Hölder's inequality.
Likewise in the case of bounded basis, we have for unbounded basis that
\begin{align*}
    B_N\lesssim\sqrt{kN}\mathbf{r}_{N,2\psi}(\mathbf{m}_{N,2\omega}+\mathbf{p}_{N,2\omega'})=o(1).
\end{align*}

\noindent\textbf{Bound on $\Lambda_N$.}
It suffices to derive the bounds on $E[S_j^2]$ for $j=1,2,3$ since
\begin{align}
    \Lambda_N\lesssim\xi_k\left(\sup_{\eta\in\mathcal{S}_N}E\left[(U-U_0)^2\right]^{1/2}\lor\sup_{\eta\in\mathcal{S}_N}E\left[(T-T_0)^2\right]^{1/2}\right),\label{eq:lambound}
\end{align}
by definition.
We have that
\begin{align*}
    E[S_1^2]=E\left[(\mu(1,X)-\mu_0(1,X))^2E\left[(1-Z\rho_0(X)^{-1})^2|X\right]\right]\lesssim\overline{C}\mathbf{m}_{N,2}^2,
\end{align*}
where the wave inequality is implied by
\begin{align*}
    E\left[(1-Z\rho_0(X)^{-1})^2|X\right]=1-2+\rho_0(X)^{-1}\leq\overline{C}.
\end{align*}
By (\ref{eq:sqbound}), we have that
\begin{align*}
    E[S_2^2]=E\left[(\rho_0(X)-\rho(X))^2E\left[(Y-\mu_0(1,X))^2|Z=1,X\right]\rho(X)^{-2}\rho_0(X)^{-1}\right]\lesssim\overline{C}^3\mathbf{r}_{N,2}^2.
\end{align*}
Lastly, $E[S_3^2]$ is bounded as
\begin{align*}
    E[S_3^2]=E\left[(\rho_0(X)-\rho(X))^2(\mu(1,X)-\mu_0(1,X))^2\rho(X)^{-2}\rho_0(X)^{-1}\right]\leq\overline{C}^5(\mathbf{m}_{N,2}^2\land\mathbf{r}_{N,2}^2),
\end{align*}
because we can bound the terms $(\rho_0(X)-\rho(X))^2$ and $(\mu(1,X)-\mu_0(1,X))^2$ by $\overline{C}^2$.
Thus, Assumption \ref{as:fslate}, (\ref{eq:lambound}) and the bounds on $E[S_j^2]$ for $j=1,2,3$ gives 
\begin{align*}
    \sup_{\eta\in\mathcal{S}_N}E\left[\|p(V)(U-U_0)\|^2\right]^{1/2}\lesssim\xi_k(\mathbf{m}_{N,2}\lor\mathbf{r}_{N,2}).
\end{align*}
Combining this result with the bound on $\sup_{\eta\in\mathcal{S}_N}E\left[\|p(V)(T-T_0)\|^2\right]^{1/2}$ gives $\Lambda_N\lesssim\xi_k(\mathbf{m}_{N,2}\lor\mathbf{p}_{N,2}\lor\mathbf{r}_{N,2})=o(1)$, which concludes the proof.
\end{proof}

\begin{proof}[Proof of Theorem \ref{th:mcrb}]
\noindent\textbf{(a)}
We have that $\mu_0(d,x)=E[Y_d|X=x]$ by Assumption \ref{as:rbcate}.1 and \ref{as:rbcate}.4.
Therefore, $\nu_0(v)=E[\mu_0(1,X)|V=v]=E[Y_1|V=v]$ and $\zeta_0(v)=E[\mu_0(0,X)|V=v]=E[Y_0|V=v]$, which together with Assumption \ref{as:rbcate}.3 indicates $\theta_0(v)=\nu_0(v)/\zeta_0(v)$.

\noindent\textbf{(b)}
The moment conditions in part (b) immediately follow from the definition of the regression functions $\mu_0$ and the propensity score $\pi_0$.

\noindent\textbf{(c)}
The Neyman orthogonality of the signals $U$ and $T$ follows from Lemma \ref{le:orthosignal}.
\end{proof}

\begin{proof}[Proof of Corollary \ref{cr:rbcate}]
The same arguments in the proof of Corollary \ref{cr:late} can be used to derive $B_N\lesssim\sqrt{kN}\mathbf{m}_{N,2}\mathbf{p}_{N,2}=o(1)$ for bounded basis, $B_N\lesssim\sqrt{kN}\mathbf{m}_{N,2\omega}\mathbf{p}_{N,2\psi}=o(1)$ for unbounded basis, and $\Lambda_N\lesssim\xi_k(\mathbf{m}_{N,2}\lor\mathbf{p}_{N,2})=o(1)$.
\end{proof}

\begin{proof}[Proof of Theorem \ref{th:mcidid}]
\noindent\textbf{(a)}
For $z=0,1$, we have that
\begin{align*}
    \mu_0(1,z,X)-\mu_0(0,z,X)&=E[D_{z1}(Y_{11}-Y_{01})+Y_{01}-D_{z0}(Y_{10}+Y_{00})-Y_{00}|Z=z,X]\\
    &=E[D_{z1}(Y_{11}-Y_{01})-D_{z0}(Y_{10}+Y_{00})|X]+E[Y_{01}-Y_{00}|X],
\end{align*}
where the first equality follows by Assumption \ref{as:idid}.1 and 3, and the second equality follows by Assumption \ref{as:idid}.5.
Therefore,
\begin{align*}
    \sum_{(w,z)\in\{0,1\}^2}(-1)^{w+z}\mu_0(w,z,X)&=E[(D_{11}-D_{01})(Y_{11}-Y_{01})-(D_{10}-D_{00})(Y_{10}+Y_{00})|X]\\
    &=E[D_{11}-D_{01}-D_{10}+D_{00}|X]E[Y_1-Y_0|X],
\end{align*}
where the second equality follows by Assumption \ref{as:idid}.6 and 7.
Moreover, we have that
\begin{align*}
    \sum_{(w,z)\in\{0,1\}^2}(-1)^{w+z}\pi_0(w,z,X)=E[D_{11}-D_{01}-D_{10}+D_{00}|X],
\end{align*}
by Assumption \ref{as:idid}.1, 3, 4 and 5.
Combining these results show that $\theta_0(v)=\nu_0(v)/\zeta_0(v)$.

\noindent\textbf{(b)}
The moment conditions in part (b) immediately follow from the definition of the regression functions $\mu_0, \pi_0$ and the propensity score $\rho_0$.

\noindent\textbf{(c)}
The Neyman orthogonality of the signals $U$ and $T$ follows from Lemma \ref{le:orthosignal}.
\end{proof}

\begin{proof}[Proof of Corollary \ref{cr:idid}]
The same arguments in the proof of Corollary \ref{cr:late} can be used to derive $B_N\lesssim\sqrt{kN}\mathbf{r}_{N,2}(\mathbf{m}_{N,2}+\mathbf{p}_{N,2})=o(1)$ for bounded basis, $B_N\lesssim \sqrt{kN}(\mathbf{m}_{N,2\omega}\mathbf{r}_{N,2\psi}+\mathbf{p}_{N,2\omega'}\mathbf{r}_{N,2\psi'})=o(1)$ for unbounded basis, and $\Lambda_N\lesssim\xi_k(\mathbf{m}_{N,2}\lor\mathbf{p}_{N,2}\lor\mathbf{r}_{N,2})=o(1)$.
\end{proof}

\begin{proof}[Proof of Theorem \ref{th:mcdc}]
\noindent\textbf{(a)}
First, we provide the proof for LATE.
We have that
\begin{align*}
    \mu_0(1,X)-\mu_0(0,X)&=E[(Z_1-Z_0)(D_1-D_0)(Y_1-Y_0)|X]\\
    &=E[Z_1-Z_0|X]E[D_1-D_0|X]E[Y_1-Y_0|D_1-D_0=1,X],
\end{align*}
where the first equality follows by Assumption \ref{as:dcid}.1 and 4, and the second equality follows by Assumption \ref{as:dcid}.2 and that $D_1-D_0\in\{0,1\}$ implied by $D_1\geq D_0$.
Likewise, we have that
\begin{align*}
    \pi_0(1,X)-\pi_0(0,X)&=E[(Z_1-Z_0)(D_1-D_0)|X]\\
    &=E[Z_1-Z_0|X]E[D_1-D_0|X].
\end{align*}
Consequently, it holds that
\begin{align*}
    \frac{\nu_0(v)}{\zeta_0(v)}&=\frac{E[Z_1-Z_0|V=v]E[D_1-D_0|V=v]E[Y_1-Y_0|D_1-D_0=1,V=v]}{E[Z_1-Z_0|V=v]E[D_1-D_0|V=v]}\\
    &=E[Y_1-Y_0|D_1>D_0,V=v]=\theta_0(v),
\end{align*}
where the first equality holds by Assumption \ref{as:dcid}.3, and the second equality holds by the fact that $D_1-D_0=1$ implies $D_1>D_0$.
Additionally, if $P(D_1>D_0|V=v)=1$ for all $v\in\mathcal{V}$, LATE is equal to CATE.
Then, the statement in Theorem \ref{th:mcdc} (a) also holds for CATE.

\noindent\textbf{(b)}
The moment conditions in part (b) immediately follow from the definition of the regression functions $\mu_0, \pi_0$ and the propensity score $\rho_0$.

\noindent\textbf{(c)}
The Neyman orthogonality of the signals $U$ and $T$ follows from Lemma \ref{le:orthosignal}.
\end{proof}

\begin{proof}[Proof of Corollary \ref{cr:dc}]
The same arguments in the proof of Corollary \ref{cr:late} can be used to derive $B_N\lesssim\sqrt{kN}\mathbf{r}_{N,2}(\mathbf{m}_{N,2}+\mathbf{p}_{N,2})=o(1)$ for bounded basis, $B_N\lesssim\sqrt{kN}(\mathbf{m}_{N,2\omega}\mathbf{r}_{N,2\psi}+\mathbf{p}_{N,2\omega'}\mathbf{r}_{N,2\psi'})=o(1)$ for unbounded basis, and $\Lambda_N\lesssim\xi_k(\mathbf{m}_{N,2}\lor\mathbf{p}_{N,2}\lor\mathbf{r}_{N,2})=o(1)$.
\end{proof}

\section{Additional Simulation Results}\label{sec:sensitivity}
Table \ref{tb:sens_dgpl} and \ref{tb:sens_dgpq} present the detailed results of the sensitivity analysis in Section \ref{sec:simdml}.
We use the common hyperparameters for the numerator and denominator of SEP, and for the target and auxiliary function of DLS.

\begin{table*}[t]
\centering
\small
\caption{Simulation Results for DGP-L with Fixed Hyperparameters}
\begin{tabular}{cccccccccccccc}
\hline
    &           &      & \multicolumn{3}{c}{$N=500$} &  & \multicolumn{3}{c}{$N=1000$} &  & \multicolumn{3}{c}{$N=2000$} \\ \cline{4-6} \cline{8-10} \cline{12-14} 
$k$ & $\lambda$ &      & Bias    & SD     & MSE    &  & Bias     & SD     & MSE    &  & Bias     & SD     & MSE    \\ \hline
3   & 0.001     & DER & 0.025   & 0.267  & 0.060  &  & 0.012    & 0.180  & 0.029  &  & 0.004    & 0.130  & 0.014  \\
    &           & SEP  & 0.392   & 0.340  & 1.843  &  & 0.389    & 0.238  & 1.605  &  & 0.389    & 0.168  & 1.497  \\
    &           & DLS  & 0.013   & 0.261  & 0.054  &  & 0.001    & 0.176  & 0.027  &  & -0.006   & 0.127  & 0.014  \\
    &           & DWLS & 0.720   & 0.549  & 0.358  &  & 0.696    & 0.373  & 0.236  &  & 0.685    & 0.271  & 0.185  \\
3   & 0.01      & DER & -0.001  & 0.257  & 0.053  &  & -0.013   & 0.173  & 0.026  &  & -0.020   & 0.125  & 0.013  \\
    &           & SEP  & 0.381   & 0.337  & 1.802  &  & 0.379    & 0.236  & 1.569  &  & 0.378    & 0.166  & 1.464  \\
    &           & DLS  & -0.082  & 0.218  & 0.052  &  & -0.086   & 0.149  & 0.038  &  & -0.089   & 0.108  & 0.030  \\
    &           & DWLS & 0.553   & 0.475  & 0.218  &  & 0.538    & 0.325  & 0.145  &  & 0.532    & 0.237  & 0.112  \\
3   & 0.1       & DER & -0.184  & 0.188  & 0.052  &  & -0.189   & 0.128  & 0.042  &  & -0.192   & 0.093  & 0.036  \\
    &           & SEP  & 0.284   & 0.310  & 1.451  &  & 0.282    & 0.217  & 1.263  &  & 0.282    & 0.153  & 1.177  \\
    &           & DLS  & -0.405  & 0.110  & 0.128  &  & -0.401   & 0.076  & 0.125  &  & -0.400   & 0.055  & 0.122  \\
    &           & DWLS & -0.117  & 0.229  & 0.080  &  & -0.118   & 0.157  & 0.069  &  & -0.117   & 0.115  & 0.064  \\
3   & 1         & DER & -0.599  & 0.058  & 0.162  &  & -0.599   & 0.040  & 0.162  &  & -0.599   & 0.029  & 0.161  \\
    &           & SEP  & -0.203  & 0.173  & 0.304  &  & -0.205   & 0.121  & 0.262  &  & -0.205   & 0.086  & 0.242  \\
    &           & DLS  & -0.754  & 0.014  & 0.226  &  & -0.754   & 0.009  & 0.226  &  & -0.753   & 0.007  & 0.225  \\
    &           & DWLS & -0.681  & 0.041  & 0.196  &  & -0.681   & 0.028  & 0.196  &  & -0.681   & 0.020  & 0.196  \\
6   & 0.001     & DER & 0.071   & 0.340  & 0.223  &  & 0.028    & 0.211  & 0.084  &  & 0.011    & 0.144  & 0.035  \\
    &           & SEP  & 0.699   & 0.495  & 1.101  &  & 0.680    & 0.347  & 0.747  &  & 0.676    & 0.244  & 0.546  \\
    &           & DLS  & 0.046   & 0.334  & 0.150  &  & 0.008    & 0.209  & 0.066  &  & -0.007   & 0.144  & 0.033  \\
    &           & DWLS & 0.623   & 0.572  & 0.551  &  & 0.563    & 0.358  & 0.273  &  & 0.539    & 0.243  & 0.173  \\
6   & 0.01      & DER & 0.043   & 0.330  & 0.176  &  & 0.003    & 0.206  & 0.070  &  & -0.013   & 0.141  & 0.031  \\
    &           & SEP  & 0.687   & 0.491  & 1.069  &  & 0.669    & 0.344  & 0.722  &  & 0.665    & 0.242  & 0.525  \\
    &           & DLS  & -0.046  & 0.295  & 0.159  &  & -0.071   & 0.195  & 0.126  &  & -0.083   & 0.139  & 0.109  \\
    &           & DWLS & 0.502   & 0.520  & 0.396  &  & 0.456    & 0.333  & 0.244  &  & 0.438    & 0.230  & 0.183  \\
6   & 0.1       & DER & -0.122  & 0.264  & 0.125  &  & -0.148   & 0.170  & 0.090  &  & -0.158   & 0.121  & 0.076  \\
    &           & SEP  & 0.583   & 0.458  & 0.819  &  & 0.568    & 0.321  & 0.531  &  & 0.565    & 0.227  & 0.369  \\
    &           & DLS  & -0.293  & 0.173  & 0.212  &  & -0.291   & 0.123  & 0.207  &  & -0.290   & 0.091  & 0.204  \\
    &           & DWLS & 0.012   & 0.311  & 0.269  &  & 0.002    & 0.208  & 0.238  &  & -0.002   & 0.148  & 0.225  \\
6   & 1         & DER & -0.499  & 0.109  & 0.196  &  & -0.503   & 0.074  & 0.192  &  & -0.504   & 0.055  & 0.190  \\
    &           & SEP  & 0.049   & 0.286  & 0.401  &  & 0.041    & 0.200  & 0.281  &  & 0.042    & 0.144  & 0.225  \\
    &           & DLS  & -0.677  & 0.043  & 0.228  &  & -0.674   & 0.030  & 0.228  &  & -0.672   & 0.022  & 0.228  \\
    &           & DWLS & -0.587  & 0.083  & 0.222  &  & -0.587   & 0.055  & 0.220  &  & -0.585   & 0.040  & 0.219  \\
10  & 0.001     & DER & -0.037  & 0.455  & 0.340  &  & -0.020   & 0.265  & 0.089  &  & -0.021   & 0.166  & 0.032  \\
    &           & SEP  & 0.362   & 0.392  & 2.407  &  & 0.360    & 0.268  & 1.819  &  & 0.349    & 0.182  & 1.504  \\
    &           & DLS  & -0.053  & 0.405  & 0.230  &  & -0.041   & 0.249  & 0.077  &  & -0.041   & 0.159  & 0.030  \\
    &           & DWLS & 0.389   & 0.865  & 1.136  &  & 0.473    & 0.493  & 0.380  &  & 0.478    & 0.311  & 0.186  \\
10  & 0.01      & DER & -0.069  & 0.419  & 0.283  &  & -0.063   & 0.242  & 0.078  &  & -0.065   & 0.153  & 0.030  \\
    &           & SEP  & 0.346   & 0.386  & 2.332  &  & 0.344    & 0.263  & 1.759  &  & 0.333    & 0.179  & 1.455  \\
    &           & DLS  & -0.160  & 0.283  & 0.125  &  & -0.160   & 0.188  & 0.076  &  & -0.162   & 0.127  & 0.053  \\
    &           & DWLS & 0.315   & 0.666  & 0.747  &  & 0.343    & 0.394  & 0.282  &  & 0.333    & 0.255  & 0.150  \\
10  & 0.1       & DER & -0.218  & 0.302  & 0.162  &  & -0.240   & 0.174  & 0.077  &  & -0.251   & 0.114  & 0.053  \\
    &           & SEP  & 0.223   & 0.341  & 1.853  &  & 0.220    & 0.231  & 1.376  &  & 0.213    & 0.158  & 1.136  \\
    &           & DLS  & -0.367  & 0.151  & 0.131  &  & -0.369   & 0.108  & 0.123  &  & -0.372   & 0.079  & 0.116  \\
    &           & DWLS & 0.025   & 0.365  & 0.297  &  & 0.005    & 0.241  & 0.160  &  & -0.013   & 0.164  & 0.114  \\
10  & 1         & DER & -0.463  & 0.140  & 0.118  &  & -0.479   & 0.090  & 0.110  &  & -0.486   & 0.064  & 0.107  \\
    &           & SEP  & -0.122  & 0.231  & 1.008  &  & -0.136   & 0.155  & 0.752  &  & -0.141   & 0.108  & 0.630  \\
    &           & DLS  & -0.631  & 0.056  & 0.164  &  & -0.626   & 0.041  & 0.159  &  & -0.622   & 0.031  & 0.155  \\
    &           & DWLS & -0.459  & 0.142  & 0.123  &  & -0.466   & 0.094  & 0.112  &  & -0.469   & 0.067  & 0.106  \\ \hline
\end{tabular}
\label{tb:sens_dgpl}
\end{table*}

\begin{table*}[t]
\centering
\small
\caption{Simulation Results for DGP-Q with Fixed Hyperparameters}
\begin{tabular}{cccccccccccccc}
\hline
    &           &      & \multicolumn{3}{c}{$N=500$} &  & \multicolumn{3}{c}{$N=1000$} &  & \multicolumn{3}{c}{$N=2000$} \\ \cline{4-6} \cline{8-10} \cline{12-14} 
$k$ & $\lambda$ &      & Bias    & SD     & MSE    &  & Bias     & SD     & MSE    &  & Bias     & SD     & MSE    \\ \hline
3   & 0.001     & DER & 0.012   & 0.318  & 0.399  &  & 0.000    & 0.215  & 0.354  &  & -0.005   & 0.158  & 0.336  \\
    &           & SEP  & 0.443   & 0.407  & 2.793  &  & 0.442    & 0.285  & 2.505  &  & 0.444    & 0.205  & 2.386  \\
    &           & DLS  & 0.004   & 0.311  & 0.383  &  & -0.007   & 0.211  & 0.343  &  & -0.011   & 0.155  & 0.327  \\
    &           & DWLS & 1.044   & 0.738  & 1.704  &  & 1.016    & 0.511  & 1.465  &  & 1.005    & 0.369  & 1.361  \\
3   & 0.01      & DER & -0.011  & 0.305  & 0.378  &  & -0.021   & 0.207  & 0.338  &  & -0.026   & 0.153  & 0.322  \\
    &           & SEP  & 0.432   & 0.403  & 2.749  &  & 0.431    & 0.283  & 2.467  &  & 0.433    & 0.203  & 2.351  \\
    &           & DLS  & -0.070  & 0.263  & 0.302  &  & -0.074   & 0.181  & 0.282  &  & -0.076   & 0.133  & 0.274  \\
    &           & DWLS & 0.840   & 0.638  & 1.214  &  & 0.822    & 0.445  & 1.072  &  & 0.817    & 0.322  & 1.011  \\
3   & 0.1       & DER & -0.178  & 0.227  & 0.277  &  & -0.182   & 0.156  & 0.261  &  & -0.184   & 0.115  & 0.255  \\
    &           & SEP  & 0.331   & 0.371  & 2.365  &  & 0.330    & 0.260  & 2.134  &  & 0.332    & 0.187  & 2.039  \\
    &           & DLS  & -0.382  & 0.133  & 0.209  &  & -0.379   & 0.093  & 0.205  &  & -0.377   & 0.068  & 0.204  \\
    &           & DWLS & 0.026   & 0.307  & 0.357  &  & 0.025    & 0.214  & 0.337  &  & 0.027    & 0.156  & 0.330  \\
3   & 1         & DER & -0.591  & 0.070  & 0.218  &  & -0.591   & 0.049  & 0.217  &  & -0.591   & 0.036  & 0.217  \\
    &           & SEP  & -0.177  & 0.208  & 0.916  &  & -0.178   & 0.145  & 0.863  &  & -0.177   & 0.105  & 0.845  \\
    &           & DLS  & -0.751  & 0.016  & 0.235  &  & -0.751   & 0.011  & 0.235  &  & -0.750   & 0.008  & 0.235  \\
    &           & DWLS & -0.657  & 0.054  & 0.228  &  & -0.656   & 0.037  & 0.228  &  & -0.656   & 0.027  & 0.228  \\
6   & 0.001     & DER & 0.095   & 0.463  & 0.427  &  & 0.037    & 0.279  & 0.168  &  & 0.016    & 0.193  & 0.066  \\
    &           & SEP  & 0.940   & 0.605  & 1.984  &  & 0.924    & 0.427  & 1.370  &  & 0.926    & 0.304  & 1.089  \\
    &           & DLS  & 0.100   & 0.448  & 0.249  &  & 0.044    & 0.276  & 0.110  &  & 0.023    & 0.192  & 0.050  \\
    &           & DWLS & 0.862   & 0.753  & 1.019  &  & 0.784    & 0.472  & 0.465  &  & 0.755    & 0.319  & 0.256  \\
6   & 0.01      & DER & 0.088   & 0.449  & 0.326  &  & 0.032    & 0.274  & 0.133  &  & 0.011    & 0.190  & 0.055  \\
    &           & SEP  & 0.928   & 0.601  & 1.949  &  & 0.912    & 0.424  & 1.348  &  & 0.914    & 0.302  & 1.073  \\
    &           & DLS  & 0.065   & 0.382  & 0.118  &  & 0.030    & 0.255  & 0.071  &  & 0.016    & 0.183  & 0.049  \\
    &           & DWLS & 0.764   & 0.691  & 0.558  &  & 0.704    & 0.446  & 0.294  &  & 0.681    & 0.305  & 0.183  \\
6   & 0.1       & DER & -0.010  & 0.361  & 0.111  &  & -0.048   & 0.231  & 0.060  &  & -0.062   & 0.166  & 0.039  \\
    &           & SEP  & 0.819   & 0.565  & 1.673  &  & 0.806    & 0.398  & 1.170  &  & 0.809    & 0.285  & 0.941  \\
    &           & DLS  & -0.189  & 0.223  & 0.090  &  & -0.185   & 0.161  & 0.079  &  & -0.181   & 0.120  & 0.073  \\
    &           & DWLS & 0.219   & 0.428  & 0.140  &  & 0.208    & 0.288  & 0.092  &  & 0.206    & 0.202  & 0.071  \\
6   & 1         & DER & -0.412  & 0.150  & 0.108  &  & -0.418   & 0.102  & 0.105  &  & -0.419   & 0.075  & 0.103  \\
    &           & SEP  & 0.235   & 0.368  & 0.874  &  & 0.227    & 0.258  & 0.676  &  & 0.232    & 0.188  & 0.593  \\
    &           & DLS  & -0.643  & 0.057  & 0.177  &  & -0.640   & 0.040  & 0.175  &  & -0.637   & 0.029  & 0.173  \\
    &           & DWLS & -0.530  & 0.116  & 0.143  &  & -0.529   & 0.078  & 0.139  &  & -0.526   & 0.055  & 0.137  \\
10  & 0.001     & DER & -0.244  & 0.605  & 1.157  &  & -0.213   & 0.320  & 0.471  &  & -0.211   & 0.200  & 0.328  \\
    &           & SEP  & 0.221   & 0.415  & 4.241  &  & 0.229    & 0.285  & 3.394  &  & 0.222    & 0.197  & 2.966  \\
    &           & DLS  & -0.221  & 0.509  & 0.769  &  & -0.203   & 0.296  & 0.431  &  & -0.205   & 0.191  & 0.322  \\
    &           & DWLS & 0.306   & 1.193  & 3.174  &  & 0.423    & 0.633  & 1.455  &  & 0.436    & 0.390  & 0.965  \\
10  & 0.01      & DER & -0.206  & 0.557  & 0.992  &  & -0.197   & 0.292  & 0.443  &  & -0.203   & 0.184  & 0.321  \\
    &           & SEP  & 0.217   & 0.409  & 4.194  &  & 0.224    & 0.280  & 3.363  &  & 0.217    & 0.194  & 2.944  \\
    &           & DLS  & -0.183  & 0.331  & 0.434  &  & -0.183   & 0.220  & 0.350  &  & -0.193   & 0.151  & 0.304  \\
    &           & DWLS & 0.416   & 0.897  & 2.286  &  & 0.441    & 0.506  & 1.234  &  & 0.420    & 0.322  & 0.882  \\
10  & 0.1       & DER & -0.149  & 0.423  & 0.619  &  & -0.184   & 0.227  & 0.369  &  & -0.207   & 0.146  & 0.300  \\
    &           & SEP  & 0.179   & 0.369  & 3.846  &  & 0.181    & 0.250  & 3.126  &  & 0.176    & 0.174  & 2.774  \\
    &           & DLS  & -0.306  & 0.180  & 0.278  &  & -0.298   & 0.131  & 0.264  &  & -0.295   & 0.099  & 0.259  \\
    &           & DWLS & 0.237   & 0.505  & 1.132  &  & 0.211    & 0.340  & 0.765  &  & 0.183    & 0.227  & 0.622  \\
10  & 1         & DER & -0.372  & 0.205  & 0.319  &  & -0.391   & 0.129  & 0.269  &  & -0.401   & 0.091  & 0.253  \\
    &           & SEP  & -0.008  & 0.283  & 2.643  &  & -0.020   & 0.192  & 2.246  &  & -0.022   & 0.137  & 2.056  \\
    &           & DLS  & -0.594  & 0.075  & 0.231  &  & -0.585   & 0.056  & 0.230  &  & -0.579   & 0.043  & 0.231  \\
    &           & DWLS & -0.356  & 0.211  & 0.362  &  & -0.365   & 0.142  & 0.316  &  & -0.370   & 0.096  & 0.301  \\ \hline
\end{tabular}
\label{tb:sens_dgpq}
\end{table*}
\end{document}